\documentclass[11pt]{article}

\usepackage{fullpage}
\usepackage{hdb_macros}
\usetikzlibrary[patterns]

\def\diagram#1#2#3{
\foreach \ifrom/\ito/\jfrom/\jto/\c in #3
  \foreach \i in {\ifrom,...,\ito}
    \foreach \j in {\jfrom, ...,\jto}{
      \ifthenelse{\equal{\c}{red}}
      {\draw[draw=none, fill=red] (\i,\j) rectangle (\i+1,\j+1);}
      {\draw[draw=none, fill=none,pattern=crosshatch] (\i,\j) rectangle (\i+1,\j+1);}
      }
\draw (0,0) grid (#1,#2);
}

\newcommand{\Poly}[2]{\mathcal{P}_{#1}(#2)}
\newcommand{\h}{\mathrm{h}}
\newcommand{\ent}{\mathrm{H}_\infty}
\newcommand{\ppoly}{\mathsf{P}/{\mathsf{poly}}}
\title{
Hilbert Functions and Low-Degree Randomness Extractors
}

 \author{
 Alexander Golovnev\thanks{Georgetown University. Email: \texttt{alexgolovnev@gmail.com}. Supported by the NSF CAREER award (grant CCF-2338730). 
}
 \and
 Zeyu Guo\thanks{The Ohio State University. Email: \texttt{zguotcs@gmail.com}.}
 \and
 Pooya Hatami\thanks{The Ohio State University. Email: \texttt{pooyahat@gmail.com}. Supported by NSF grant CCF-1947546.} 
 \and
 Satyajeet Nagargoje\thanks{Georgetown University. Email: \texttt{satyajeetn2012@gmail.com}. Supported by the NSF CAREER award (grant CCF-2338730).}
 \and
 Chao Yan\thanks{Georgetown University. Email: \texttt{cy399@georgetown.edu}.}
 }
\date{}

\begin{document}
\maketitle

\begin{abstract}
For $S\subseteq \F^n$, consider the linear space of restrictions of degree-$d$ polynomials to~$S$. The~Hilbert function of $S$, denoted~$\h_S(d,\F)$, is the dimension of this space. We obtain a tight lower bound on the smallest value of the Hilbert function of subsets $S$ of arbitrary finite grids in $\F^n$ with a fixed size $|S|$. We achieve this by proving that this value coincides with a combinatorial quantity, namely the smallest number of low Hamming weight points in a down-closed set of size $|S|$. 

Understanding the smallest values of Hilbert functions is closely related to the study of degree-$d$
closure of sets, a notion introduced by Nie and Wang (Journal of Combinatorial Theory, Series A, 2015). We use bounds on the Hilbert function to obtain a tight bound on the size of degree-$d$ closures of subsets of $\F_q^n$, which answers a question posed by Doron, Ta-Shma, and Tell (Computational Complexity, 2022).

We use the bounds on the Hilbert function and degree-$d$ closure of sets to prove that a random low-degree polynomial is an extractor for samplable randomness sources. Most notably, we prove the existence of low-degree extractors and dispersers for sources generated by constant-degree polynomials and polynomial-size circuits. Until recently, even the existence of arbitrary deterministic extractors for such sources was not known.
\end{abstract}
\pagenumbering{roman}
\thispagestyle{empty}
\newpage
\setcounter{tocdepth}{3}
\tableofcontents
\newpage
\pagenumbering{arabic}

\section{Introduction}

\paragraph{Hilbert Functions.}
Low-degree polynomials are fundamental objects in theoretical computer science, and their properties are extensively studied due to their important role in areas such as error correcting codes and  circuit lower bounds. Let $d\geq 0$ be an integer, $\F$ be a field, and $S\subseteq \F^n$ be a set. Each degree-$d$ $n$-variate polynomial $p$ over $\F$ can  be naturally viewed as a map $p:\F^n\rightarrow \F$, and hence also defines a map $p\vert_S: S\rightarrow \F$. Considering the linear space of all such maps in $\F^S$, which is a subspace of the space of all maps from $S$ to $\F$, allows one to tap into a wide array of algebraic techniques in order to prove useful facts about the set $S$. This approach was for example utilized in complexity theory famously in the work of Smolensky~\cite{Smolensky87}, where proving bounds on the dimension of the aforementioned subspace was used to obtain lower-bounds for $\AC^0[\oplus]$ circuits computing the indicator function of the set $S$, for various $S\subseteq \{0,1\}^n$. The dimension of the space consisting of $p|_S$ for all degree-$d$ polynomials $p$ is indeed a well-studied and classical concept in algebraic geometry known as the (affine) Hilbert function of $S$, denoted by $\h_S(d, \F)$. 
Hilbert functions encode important geometric and algebraic information, such as the dimension, degree, and regularity of varieties, in a more general context.

Hilbert functions have previously been studied in complexity theory due to their applications in circuit lower bounds, in particular for $\AC^0[\oplus]$ circuits, that were established by Smolensky~\cite{Smolensky87} and Razborov~\cite{Razborov87}. Such applications require finding sets $S\subseteq \{0,1\}^n$ where the Hilbert function takes a very large value. However, it is also interesting to prove general lower bounds or find lower-bounding methods for arbitrary sets $S$. An example of such a result is the work of Moran and Rashtchian~\cite{moran2016shattered}, who showed upper and lower bounds on $\h_S(d,\F)$ for $S\subseteq \{0,1\}^n\subseteq \F^n$ via various concepts in VC theory. \cite{moran2016shattered} treated the Hilbert function as a complexity measure of the set~$S$ and compared it to measures that arise naturally in learning theory, including ``shattering'' and ``strong shattering'' values. 

Suppose $r>0$ is an integer. It is natural to wonder, what the extreme values of $\h_S(d,\F)$ are, among all sets $S$ of size $|S|=r$. It is not hard to show that the maximum value is equal to $\min\left( r, \h_{\F^n}(d,\F)\right)$ when $S\subseteq \F^n$. For example, the maximum value of the Hilbert function of a set $S\subseteq \F_2^n$ of size $r$ is $\min(r, \binom{n}{\leq d})$.  

On the other hand, finding the true smallest value of $\h_S(d,\F)$ is a natural and intriguing question even when $S$ is restricted to subsets of some finite and structured set in~$\F^n$.
\begin{question}\label{question:hilbert}
Let $0\leq d \leq n$ be integers, $\F$ be a field, and $A=A_1\times \cdots \times A_n\subseteq \F^n$ where $A_i\subseteq \F$ are finite sets. For any $r\leq |A|$, what is the smallest value of $\h_S(d, \F)$ among all subsets $S\subseteq A^n$ of cardinality $|S|=r$?
\end{question}
This question has been answered in the case of $\F=\F_2$ and $A=\F_2^n$ by Keevash and Sudakov~\cite{keevash2005set} and Ben-Eliezer, Hod, and Lovett~\cite{ben2012random}, and later generalized to $\F=\F_p$ and $A=\F_p^n$ by Beame, Oveis Gharan, and Yan~\cite{beame2020bias}. For simplicity, let $r=p^k$ for some $k\geq 0$. \cite{beame2020bias} proved that the smallest value of $\h_S(d,\F_p)$ with $|S|=r$ is equal to the number of degree-$\leq d$ monomials on $k$ variables, for example when $p=2$, this is equal to simply $\binom{k}{\leq d}=\binom{\log_2 r}{\leq d}$. 

We prove a more general result that answers \cref{question:hilbert}  for arbitrary finite grids $A\subseteq \F^n$ in arbitrary fields $\F$. We show that the smallest values of Hilbert functions are exactly determined by an extremal combinatorial question about the number of low-Hamming-weight elements in down-closed sets, which we solve by building on the work of Beelen and Datta~\cite{BD18}.

The prior works discussed above were motivated by applications in bounding the list-size of the Reed-Muller codes and obtaining certain extensions of Frankl–Ray-Chaudhuri–Wilson theorems on cross-intersecting sets. In contrast, in this paper, we are interested in \cref{question:hilbert} due to its applications in pseudorandomness, particularly in randomness extraction.


Understanding the smallest values of Hilbert functions is closely related to the study of \emph{degree-$d$ closure} of sets, a notion introduced by Nie and Wang \cite{nie2015hilbert}.

\begin{definition}\label{def:Span}
The degree-$d$ closure of a set $T\subseteq \F^n$ is defined as 
\[
\cl_d(T)\coloneqq \{a \in \F^n\ | \ \text{ for every degree-$d$ polynomial $f$, }  f |_T\equiv 0 \Rightarrow  f(a)=0\} \;.
\]
\end{definition}
Equivalently, $\cl_d(T)$ is the set of all points $a\in \F^n$ such that the values $f(a)$ of a degree-$d$ polynomial~$f$ are determined by $f|_T$.

The existence of a small set with a large degree-$d$ closure has application to hitting-set generators for polynomials \cite{DTT22}. As an application of our answer to \cref{question:hilbert}, we obtain an upper bound on the size of $\cl_d(T)$ in terms of $|T|$. Our bound in fact yields an optimal way of creating a small set with a large degree-$d$ closure. 



Futhermore, \cref{question:hilbert} has direct implications to the theory of randomness extractors, which we discuss next.

\paragraph{Randomness Extractors.}
The theory of randomness extractors is an active research area that was initiated in \cite{SANTHA198675, blum1986independent} with the motivation of simulating randomized algorithms with access to ``weak'' randomness sources. 
The main objective of this theory is to design \emph{extractors} that are capable of purifying imperfect randomness sources into high-quality random bits or bit sequences. Extractors and related objects such as \emph{dispersers, samplers, and condensers} have since found numerous applications in constructing other pseudorandom objects such as pseudorandom generators~\cite{NISAN199643} and expander graphs~\cite{WZ99}, as well as applications in other areas of theoretical computer science and mathematics including cryptography~\cite{dodis2004possibility}, combinatorics~\cite{li2023two}, hardness of approximation~\cite{v003a006}, error correcting codes~\cite{ta2001extractor}, and metric embeddings~\cite{indyk2007uncertainty}. 

A deterministic extractor for a family $\mathcal{X}$ of distributions over $\{0,1\}^n$ is a map $f:\{0,1\}^n\rightarrow \{0,1\}^m$ such that for any $\mathbf{X}\in \mathcal{X}$, $f(\mathbf{X})$ is close to the uniform distribution in statistical distance. It is common to measure the amount of randomness in a random variable $\mathbf{X}$ by its min-entropy, defined $\ent(\mathbf{X})\coloneqq -\log_2 \max_{x\in \{0,1\}^n} \Pr[\textbf{X}=x]$. It is easy to show that no deterministic extractor can extract from general $n$-bit randomness sources of min-entropy as high as $n-1$~\cite{chor1988unbiased}. As a result, researchers in the area have explored two directions. Much of the focus in the area has been given to the more powerful \emph{seeded extractors} that have access to an additional short purely random seed. This article contributes to another line of work that has extensively investigated the extra assumptions on the randomness sources that allow for explicit deterministic extractors and dispersers to exist. A widely studied class of sources in this latter direction, introduced in \cite{MR1931802} is ``samplable sources'', where the sources of randomness are distributions sampled by applying a low-complexity map (e.g., a decision forest, local map, $\NC^0$ circuit, $\AC^0$ circuit, an affine or a low-degree map) to the uniform distribution.
Unfortunately, constructing explicit extractors even for sources samplable by really low-complexity maps has been quite challenging, and for example all the known constructions of extractors for local sources require quite high min-entropy of $\Omega(\sqrt{n})$~\cite{viola2014extractors,DW11localextractors}. Due to the difficulty of constructing good explicit extractors and motivated by applications in complexity theory such as circuit lower bounds~\cite{li20221} and lower bounds for distribution-sampling~\cite{viola2012complexity}, researchers have considered the seemingly easier task of proving the existence of low-complexity extractors~\cite{viola2005complexity,goldreich2015randomness,cheng2018randomness,vadhan2004constructing,bogdanov2013sparse,dodis2021doubly,huang2022affine,cohen2015two,alrabiah2022low}. 

The state of affairs is much worse when it comes to randomness extraction from sources sampled by more powerful maps such as $\AC^0[\oplus]$ or low-degree $\F_2$-polynomial maps. In this case obtaining nontrivial explicit constructions and even non-explicit low-complexity extractors remains open. In fact, the same problems are open even in the case of dispersers. Here a map $f:\{0,1\}^n \rightarrow \{0,1\}$ is a \emph{disperser} for $\mathcal{X}$ if for every source $\mathbf{X}\in \mathcal{X}$, the support of $f(\mathbf{X})$ is $\{0,1\}$. On the positive side, Chattopadhyay, Goodman, and Gurumukhani~\cite{chattopadhyay2024extractors}, recently proved the existence of deterministic (not necessarily low-complexity) extractors for low-degree $\F_2$-polynomial sources with logarithmic min-entropy.

\subsection{Our Results on Hilbert Functions} 
We obtain our answer to \cref{question:hilbert} by first reducing it to a purely combinatorial problem. In particular, via an algebraic geometric argument, we prove the following theorem which states that the minimum value of Hilbert functions over subsets of a grid is exactly captured by a combinatorial quantity related to down-closed sets. (A set $T\subseteq \N^n$ is said to be down-closed if $T$ is closed under decreasing any coordinates of its elements.)

\begin{theorem}[See \cref{cor:characterization}]\label{prop:hilbert}
Let $\F$ be a field, and $A_1,\dots, A_n\subseteq \F$ be finite sets of size $|A_i|=r_i$. Define $A=A_1\times \cdots \times A_n$. For every $k\leq |A|$, 
\[
\min_{S\subseteq A: |S|=k} \h_S(d,\F)= \min_{\textup{down-closed } T\subseteq F: |T|=k} \left|T_{\leq d}\right|\;,
\]
where $F=\prod_i \{0,\dots, r_i-1\}$ and $T_{\leq d}=\{x\in T: \sum_i x_i \leq d\}$. 
\end{theorem}

Let $I$ be the ideal of $\F[X_1,\ldots,X_n]$ associated with a set $S\subseteq A$, that is, the set of all polynomials vanishing on~$S$.
Classical results in algebraic geometry (such as Hilbert's Nullstellensatz) establish close connections between the structure of~$S$ and the structure of~$I$, which allows us to focus on studying~$I$.

The proof of \cref{prop:hilbert} is based on the idea that the ideal $I$ can be reasonably approximated by another ideal, \emph{the ideal of leading terms of $I$}. This approximation preserves important information about $I$, and consequently, about $S$ as well.
In particular, when the ideal of leading terms of $I$ is defined with respect to a specific total order of monomials compatible with the total degree, it can be shown that such an approximation preserves the value of the Hilbert function. One advantage of working with the ideal of leading terms is that it is a \emph{monomial ideal}, that is, an ideal generated by monomials, whose relatively simple structure can be analyzed using combinatorial tools.

We remark that the concept of transforming a general ideal into a monomial ideal is closely related to the theory of Gr\"obner bases, which serves as a basis of computational algebraic geometry. For a detailed discussion, see, e.g., \cite{BM93}. This concept is also used in Smolensky’s algebraic method for proving circuit lower bounds \cite{Smolensky93}.

\cref{prop:hilbert} allows us to reduce the problem of determining the smallest value of Hilbert function of a set of size $k$ to understanding the smallest number of low-Hamming-weight points in down-closed sets of the same size. We then solve this combinatorial problem by proving that the minimum is obtained by the down-closed set $\mathrm{M}_F(k)$ which is defined as the set of $k$ lexicographically first elements of $F$. 
\begin{theorem}[See \cref{thm:hilbertbound}]
Let $1\leq r_1\leq \cdots\leq r_n$ be integers and let $F=\prod_{i=1}^n \{0,\ldots, r_i-1\}$.  Then 
\[
\min_{\textup{down-closed } T\subseteq F} \left|T_{\leq d}\right| = |\mathrm{M}_F(k)_{\leq d}|\;. 
\]
\end{theorem}

In the case of $r_1=\cdots=r_n=2$, we prove the above theorem via an elementary combinatorial argument, that via a series of operations turns any set of $k$ elements into $M_F(k)$ without increasing the number of elements of Hamming weight $\leq d$. We prove the general case by building on a recent result of Beelen and Datta~\cite{BD18}. This result generalizes the work of Wei \cite{Wei91} and Heijnen--Pellikaan \cite{HP98, Hei99} in studying the generalized Hamming weights of certain linear codes.

We record the following corollary of our results specialized to finite fields which generalizes the bounds due to \cite{keevash2005set,ben2012random,beame2020bias}.
\begin{corollary}[See \cref{cor:boundsOnHilbert}]\label{thm:intro-main-hilbert-1}
For every prime power $q$, and $n,k,d\in \N$ where $k\leq q^n$, we have 
    \[
    \min_{S\subseteq \F_q^n: |S|=k} \h_S(d,\F_q)= |\mathrm{M}^n_q(k)_{\leq d}|\;.
    \]
In particular, setting $q=2$, for every $n,k,d\in\N$ where $k\leq2^n$, and every $S\subseteq\F_2^n$ of size $|S|=k$,
\[
    \h_S(d,\F_2)\geq \binom{\floor{\log(k)}}{\leq d}\;.
\]
\end{corollary}

\subsubsection{\texorpdfstring{Degree-$d$}{Degree-d} Closure of Sets}
Motivated by its applications to combinatorial geometry, the notion of degree-$d$ closures of subsets of~$\mathbb{F}_q^n$ was introduced in \cite{nie2015hilbert}. This concept has since found further applications and connections to complexity theory \cite{SS18,OSS19,Sri23} and pseudorandomness \cite{DTT22}.

Recall that the degree-$d$ closure $\cl_d(T)$ of a set $T\subseteq \F^n$ over a finite field $\F$ is the set of all points $a\in\F^n$ such that any degree-$d$ polynomial vanishing on $T$ also vanishes at $a$. 
Nie and Wang \cite{nie2015hilbert} proved the following result.
\begin{theorem}[{\cite[Theorem~5.6]{nie2015hilbert}}]\label{thm:Nie-Wang}
Let $n,d\in\N$ and $T\subseteq\F_q^n$. Then
\[
|\cl_d(T)| \leq \frac{q^n}{\h_{\F_q^n}(d, \F_q)} \cdot |T|. 
\]
\end{theorem}

Building on our results on Hilbert functions, we obtain an improvement of \cref{thm:Nie-Wang} by obtaining a tight upper bound on the size of degree-$d$ closures of sets. 
\begin{theorem}[See \cref{thm:tight-bound-cl} and \cref{thm:closure-tightness}]\label{introthm:tight-bound-cl}
Let $n,d,m\in\N$. Let $T\subseteq\F_q^n$ be a set of size~$m$.
Then
\begin{equation}\label{eq:closure-size-intro}
|\cl_d(T)|\leq \max_{0\leq k\leq q^n: |\mathrm{M}^n_q(k)_{\leq d}|\leq m} k = \begin{cases}
\max_{0\leq k\leq q^n: |\mathrm{M}^n_q(k)_{\leq d}|=m} k & \text{if } m\leq \h_{\F_q^n}(d, \F_q),\\
q^n & \text{otherwise.}
\end{cases}
\end{equation}
Moreover, this bound is tight in the sense that for any $0\leq m\leq q^n$, there exists $T\subseteq\F_q^n$ of size $m$ for which \eqref{eq:closure-size-intro} holds with equality.
\end{theorem}

In fact, the set $T$ of size $m$ that attains the bound in the above theorem can be constructed explicitly; see \cref{thm:closure-tightness} for details.

For convenience, we state the following corollary which is used later in the paper. For $n,d,\delta\in\N$, denote by $N(n,d,\delta)$ the number of monomials $X_1^{e_1}\cdots X_n^{e_n}$ with $e_1,\dots,e_n\leq \delta$ and $e_1+\dots+e_n\leq d$.

\begin{corollary}\label{cor:bound-span}
Let $n,d,\ell\in\N$. If $T\subseteq\F_q^n$ is a set of size less than $N(\ell,d,q-1)$, then $|\cl_d(T)| < q^\ell$.
In particular, if $q=2$ and $T\subseteq\F_2^n$ is a set of size less than $\binom{\ell}{\leq d}$, then $|\cl_d(T)| < 2^\ell$. 
\end{corollary}
\begin{proof}
Observe that $|\mathrm{M}^n_q(q^\ell)_{\leq d}|=N(\ell,d,q-1)$. Then apply \cref{introthm:tight-bound-cl}.
\end{proof}

Let us compare our bound with the bound of Nie and Wang in some specific settings.

\begin{example}
Suppose $\ell\leq n$. Let $T\subseteq\F_2^n$ be a set of size $\binom{\ell}{\leq d}-1$.
Then by \cref{cor:bound-span}, we have the bound $|\cl_d(T)|\leq 2^\ell-1$.
On the other hand, the bound of Nie and Wang (\cref{thm:Nie-Wang}) gives
\[
|\cl_d(T)|\leq \frac{2^n}{\binom{n}{\leq d}} \cdot |T| \;,
\]
which is exponential in $n$, rather than in $\ell$, at least when $d\leq \left(\frac{1}{2}-c\right) n$ for some constant $c>0$.
\end{example}

\begin{example}
Suppose $\ell\leq n$ and $d<q$. Let $T\subseteq\F_q^n$ be a set of size $N(\ell,d,q-1)-1=\binom{\ell+d}{d}-1$.
Then by \cref{cor:bound-span}, we have the bound $|\cl_d(T)|\leq q^\ell-1$, which is exponential in $\ell \log q$.
On the other hand, the bound of Nie and Wang (\cref{thm:Nie-Wang}) gives
\[
|\cl_d(T)|\leq \frac{q^n}{\binom{n+d}{d}} \cdot |T|\;,
\]
which is exponential in $n\log q$, rather than in $\ell \log q$, at least when $n+d\leq q^{1-c}$ for some constant $c>0$.
\end{example}

In \cite{DTT22}, Doron, Ta-Shma, and Tell explicitly asked if there exists a small set $T
\subseteq \F_q^n$ whose degree-$d$ closure is very large. Our \cref{introthm:tight-bound-cl} gives an upper bound on the size of the degree-$d$ closure of $T$ in terms of the size of $T$, which is tight in the sense that there exist sets $T$ that exactly meet this bound for every cardinality of $T$. Moreover, such sets $T$ can be constructed explicitly (see \cref{thm:closure-tightness}). Thus, we completely resolve the question posed by Doron, Ta-Shma, and Tell.


\subsection{Our Results on Randomness Extractors}
Continuing the line of work on low-complexity extractors, in this paper we investigate the power of low-degree polynomials in randomness extraction. 
\begin{question}\label{q:lowdegextraction}
For which families $\mathcal{X}$ of sources does there exist a low-degree disperser? Similarly, for which families $\mathcal{X}$ of sources does there exist a low-degree extractor? 
\end{question}
Let us first discuss the easier task of obtaining low-degree dispersers before moving on to our main application of low-degree extractors. For simplicity, we will focus on the most important case of extracting randomness over $\F_2$, but all our results easily generalize to $\F_q$. Non-explicit constructions of low-degree dispersers can be obtained via understanding the probability that a random low-degree polynomial is a disperser for a family $\mathcal{X}$ of distributions over $\{0,1\}^n$ which we identify with $\F_2^n$ in the natural way. Our starting point is the observation that the notion of Hilbert functions can be used to exactly describe the probability that a random degree-$d$ polynomial $f:\{0,1\}^n \rightarrow \{0,1\}$ is a disperser for a fixed source $\mathbf{X}\in \mathcal{X}$. Indeed, this probability is exactly equal to $1-2^{1-\h_S(d, \F_2)}$, where $S=\mathrm{support}(\mathbf{X})$. Thus, in particular, \cref{thm:intro-main-hilbert-1} can be used to bound the probability that a random degree-$d$ polynomial is not a disperser for a fixed source.

\subsubsection{Low-Degree Dispersers} Let $\mathcal{X}$ be a family of sources of min-entropy $\geq k$. Observing that the support of any distribution $\mathbf{X}\in \mathcal{X}$ is of size $\geq 2^k$, one gets as an immediate corollary of \cref{thm:intro-main-hilbert-1}, the existence of low-degree dispersers $\mathcal{X}$ as long as $|\mathcal{X}|$ is small. 
\begin{theorem}[Informal, see \cref{cor:disp}]\label{thm:cor-disp-small}
Let $n,d,k\geq1$. Let $\mathcal{X}$ be a family of distributions of min-entropy $\geq k$.
Then a random degree-$d$ polynomial over $\F_2$ is a disperser for $\mathcal{X}$ with probability at least 
\[
1-|\mathcal{X}|\cdot  2^{1-{\binom{k}{\leq d}}} \;. 
\]
\end{theorem}
This theorem itself implies the existence of low-degree dispersers for several interesting families of samplable sources such as sources sampled by local maps, bounded-depth decision forests, and polynomial-sized bounded-fan-in circuits, to name a few.

A map $f:\{0,1\}^m\rightarrow \{0,1\}^n$ is called $\ell$-local if each of its output bits depends on at most~$\ell$ input bits. A depth-$\ell$ decision forest is a map $f$ where each output bit can be computed as a depth-$\ell$ decision tree. It is easy to obtain an upper bound exponential in $\poly(n)$ on the number of local or decision forest sources. Hence we get the following as a corollary of \cref{thm:cor-disp-small}.
\begin{corollary}[Informal, see \cref{cor:localdisperser}]
Let $1\leq \ell\leq d\leq n$ be integers. There exists a degree-$d$ disperser 
\begin{itemize}
    \item for the family of $\ell$-local sources on $\{0,1\}^n$ with min-entropy $k> d(2^\ell n + 2\ell n \log n)^{1/d}$. 
    \item for the family of depth-$\ell$ decision forest sources on $\{0,1\}^n$ with min-entropy $k> d((\ell+\log n) 2^{\ell+1} n)^{1/d}$. 
\end{itemize}
\end{corollary}

As mentioned above, since in addition to the min-entropy requirement, the only requirement in \cref{thm:cor-disp-small} about the family $\mathcal{X}$ is a bound on $|\mathcal{X}|$, it can be used to immediately obtain low-degree dispersers for various other families of sources as well. For example, since for any $c$, the number of Boolean circuits with $\leq n^c$ bounded fan-in gates is at most $2^{O(n^{c+1})}$, one can also use \cref{thm:cor-disp-small} to obtain a degree-$O(c)$ disperser for such families of circuits. However, we will not do an exhaustive search for all such applications, and instead our main disperser applications will focus on two powerful families of sources, namely sources sampled by low-degree polynomials over $\F_2$ and $\AC[\oplus]$ circuits which we define as the family of unbounded-depth polynomial-size Boolean circuits with AND, OR, XOR, NOT gates of unbounded fan-in, where the input gates are not counted towards the size. 

Note that low-degree polynomial maps $f:\{0,1\}^m \rightarrow \{0,1\}^n$, even affine ones, can depend on the entire input for any $m \gg n$ and thus one cannot simply bound $|\mathcal{X}|$ when $\mathcal{X}$ is the family of sources sampled by low-degree polynomials. This property holds for $\AC[\oplus]$ circuits, as we allow them to non-trivially depend on an arbitrary number of input gates (since the circuit gates have unbounded fan-in). Nevertheless, utilizing an ``input-reduction'' trick of \cite{chattopadhyay2024extractors} which applies to both the foregoing families of sources, it can be shown that for our disperser purposes we may assume the input of both families of sources to be of length $O(n)$. This allows us to apply \cref{thm:cor-disp-small} to obtain low-degree dispersers for both of these families. 

\begin{theorem}[Informal, see \cref{thm:dispPolys,thm:dispCkts}]
For every $1\leq \ell< d\leq n$, there exists a degree-$d$ disperser 
\begin{itemize}
    \item for the family of degree-$\ell$ sources on $\{0,1\}^n$ with min-entropy $k\geq(12^\ell\cdot d^d\cdot n)^{\frac{1}{d-\ell}}+1$. 
    \item for the family of $n^\ell$-size $\AC[\oplus]$ circuit sources on $\{0,1\}^n$ with min-entropy ${k\geq (30^2\cdot d^d\cdot n^{2\ell})^{\frac{1}{d-2}}+1}$.
\end{itemize}
\end{theorem}
In particular, for every $\ell\in\N$, there is a degree-$(\ell+2)$ disperser for degree-$\ell$ sources on $\{0,1\}^n$ with min-entropy $\Omega\left(\sqrt{n}\right)$.

We note that both of these source families are very powerful, and to the best of our knowledge, no nontrivial low-complexity dispersers for either of these families of sources was known prior to this work (except in the easier case of degree-$1$ sources which corresponds to affine sources for which explicit extractors for logarithmic entropy was recently proved~\cite{li2023two}). Let us also point out that the two foregoing classes have incomparable power, and that it is straightforward to use our proof technique to conclude the same result for a class of sources that generalizes both $\ACp$ and constant-degree polynomials. Indeed, the input-reduction and counting idea works for the ``hybrid'' class of polynomial-size circuits which extends $\ACp$ by allowing additional unbounded fan-in gates computing arbitrary polynomials of fixed constant degree. However, for ease of exposition, we have chosen to present only the results for $\ACp$ and low-degree sources separately.

\subsubsection{Low-Degree Extractors}
Next, we move on to another application concerning the existence of low-degree extractors for samplable sources. Can we prove the existence of low-degree extractors for all the families for which we proved the existence of low-degree dispersers? We prove this by showing an analogue of \cref{thm:cor-disp-small} for extractors.
\begin{theorem}[Informal, see \cref{thm:extr} for the more general statement]\label{thm:intro-extractor} 
Let $\mathcal{X}$ be a family of distributions of min-entropy $k\geq5\log{n}$ over $\{0,1\}^n$ for large enough~$n$. 
Then for every $d\geq 6$, a uniformly random degree-$d$ polynomial is an $\eps$-extractor for $\mathcal{X}$ with probability at least 
\[
1-|\mathcal{X}|\cdot e^{3n-O(k^{d/2})/n^2}\;
\]
for $\eps=(2d)^d\cdot k^{-d/4}$.
\end{theorem}
A similar statement (see \cref{thm:extr}) holds for families of sources that are close to convex combinations of another small family of sources. Combined with the input-reduction trick, we obtain as a corollary, the existence of low-degree extractors for various families of sources, notably, lower-degree sources and $\ACp$ circuits. 
\begin{theorem}[Informal, see \cref{thm:extrForSources}]\label{thm:intro-extrForSources}
For all $\ell,d\geq1$, and all large enough~$n$, and $k\geq 5\log n$. There exists a degree-$d$ $\F_2$-polynomial that is an $\eps$-extractor for the following families of sources over $\{0,1\}^n$ for $\eps=(2d)^d\cdot k^{-d/4}$:
\begin{itemize}
\item $\ell$-local sources for $k\geq (2^\ell n^3\log{n})^{2/d}$.
\item depth-$\ell$ decision forest sources for $k\geq (2^\ell n^3(\log{n}+\ell))^{2/d}$.
\item degree-$\ell$ sources for $k\geq(3^\ell n)^{\frac{6}{d-2\ell}}$.
\item $n^\ell$-size $\ACp$ circuit sources (with unbounded number of input gates) for $k\geq 3n^{\frac{4(\ell+1)}{d-4}}$.
\end{itemize}
\end{theorem}

In \cref{thm:extr-multi-output}, we further extend our low-degree extractors to multi-output extractors that output $\Theta(k)$ bits. This is done by independently picking random degree-$d$ polynomials $p_1,\ldots, p_t$ for some $t=\Theta(k)$, and analyzing the probability that each $p_i$ is an extractor for the family of sources obtained by $\mathcal{X}$ conditioned on the values of $p_1,\ldots, p_{i-1}$. 

Let us now discuss our proof technique for \cref{thm:intro-extractor}. Recall that \cref{thm:cor-disp-small} was a corollary to \cref{thm:intro-main-hilbert-1} which showed that a random polynomial is with high probability non-constant on the support of any fixed high min-entropy distribution. A priori it is not clear how to use this bound on the Hilbert function to prove \cref{thm:intro-extractor}. 

Indeed, let us consider the simpler case of a fixed $k$-flat source $\mathbf{X}$ over $\{0,1\}^n$, which is uniformly distributed over a set $S\subseteq \{0,1\}^n$ with $|S|=2^k$. Note that a map $p:\{0,1\}^n\rightarrow \{0,1\}$ is an extractor for $\mathbf{X}$ if it has small bias on $S$. Thus, for example, to prove the special case of \cref{thm:intro-extractor} for small families of $k$-flat sources, we would need to prove that a random degree-$d$ polynomial is small-biased on $S$ with high probability. However, \cref{thm:intro-main-hilbert-1} only tells us that $\h_S(d,\F_2)\geq \binom{k}{\leq d}$, which is not enough to prove concentration bounds for the bias of a random degree-$d$ polynomial on an arbitrary set $S$. We note that when $S$ is highly structured, that is when it is an affine subspace, this problem is equivalent to questions about list-decoding size of Reed-Muller codes, and known results such as one by Kaufman, Lovett, and Porat~\cite{KLP} that show that the number of distinct $\eps$-biased degree-$d$ polynomials on a $k$-dimensional subspace $S$ is at most $(1/\epsilon)^{k^{d-1}}$ could be utilized. However, for our application we have to deal with arbitrary sets $S$. 

\paragraph{Uniform covering by sets of full Hilbert dimension.} We say that a set $T\subseteq \{0,1\}^n$ has ``full Hilbert dimension'' if $\h_T(d, \F_2)=|T|$. Note that when $T$ has full Hilbert dimension, then the restriction of a random degree-$d$ polynomial to $T$ is uniformly distributed over $\{0,1\}^T$. In particular, if $T$ is a sufficiently large set of full Hilbert dimension, then a random degree-$d$ polynomial is small-biased on $T$ with high probability. We use this observation to design our technique for bounding the probability that a random degree-$d$ polynomial is small-biased on any fixed source $\mathbf{X}$ of large min-entropy. For simplicity we describe the idea for flat sources. In this case, $\mathbf{X}$ is uniformly distributed over a set $S$ with $|S|\geq 2^k$. It is sufficient to prove the existence of an almost-uniform covering of $S$ by large sets $T_1,\ldots, T_t$ of the same size with full Hilbert dimensions, where we call a covering almost-uniform if each element $x\in S$ belongs to roughly $tm/|S|$ many sets, where we assume $|T_i|=m$. 

We obtain such a covering by analyzing the probability that a uniformly picked subset $T_i\subseteq S$ has full Hilbert dimension. Using our bound on the Hilbert function, \cref{thm:intro-main-hilbert-1}, which allows us to bound the size of the ``degree-$d$ closure'' of small sets, we prove that a random set $T_i$ of size $m$, for some $m=\binom{\Theta(k)}{\leq d}$, has full Hilbert dimension with high probability. Similarly, we prove using the Bayes rule, that we may pick these good sets $T_i$'s of full Hilbert dimension in a way that leads to an almost uniform covering. Since $T_i$'s are of sufficiently large size $\binom{\Theta(k)}{\leq d}$ and of full Hilbert dimension, we can use the Hoeffding inequality to bound the probability that a random degree-$d$ polynomial is biased on a $T_i$ to be exponentially small in $\Theta(k)^d$, which is good enough for our applications to existence of low-degree extractors. We obtain the following result which can be used to prove \cref{thm:intro-extractor}. 

\begin{theorem}[Informal, see \cref{thm:main-extractor}]\label{thm:intro-extractor-fixedsource}
Let $n,d,k\geq1$, and $\eps>0$ be a real. Then for every distribution $\mathbf{X}$ over $\{0,1\}^n$ with $\ent(\mathbf{X})\geq k$,  a uniformly random degree-$d$ polynomial $f$ is an $\eps$-extractor for $\mathbf{X}$, with probability at least $1-e^{3n-\eps^2\binom{\ell}{\leq d}/(Cn^2)}$ where $\ell=k/2-\log(32n/\eps)$ and $C=7\cdot(32)^2$.
\end{theorem}

We find our technique of obtaining almost uniform coverings with sets of full Hilbert dimension to be powerful, and hope that it will find other applications beyond the ones explored here.

\subsection{Remarks}
\paragraph{Correlation bounds over arbitrary subsets.} We note that our proof of \cref{thm:intro-extractor-fixedsource} (\cref{thm:extr}) can be modified to the following correlation bounds with any fixed function. 

\begin{theorem}\label{thm:intro-correlationbound}
Let $n,d,k\geq1$, $\eps>0$ be a real, and $g:\F_2^n \rightarrow \F_2$ be a fixed function. Then for every distribution $\mathbf{X}$ over $\{0,1\}^n$ with $\ent(\mathbf{X})\geq k$,  for 
a uniformly random degree-$d$ polynomial $f$ we have 
\[
\Pr[f(\mathbf{X})= g(\mathbf{X})] = \frac{1}{2}\pm \eps,
\] 
with probability at least $1-e^{3n-\eps^2\binom{\ell}{\leq d}/(Cn^2)}$ where $\ell=k/2-\log(32n/\eps)$ and $C=7\cdot(32)^2$.
\end{theorem}

This generalization is quite straightforward, as once we obtain a uniform covering by sets of maximum Hilbert dimension, then Hoeffding bounds can be used to bound the correlation of a random polynomial with the fixed function restricted to the sets belonging to the cover. This can then be used to bound the over-all correlation with the fixed function in a similar way to the proof of \cref{thm:extr}.

\paragraph{Punctured Reed-Muller codes.}
The special case of \cref{thm:intro-correlationbound} when $\mathbf{X}$ is a flat source can be interpreted as a bound on the list-decoding size of Reed-Muller codes when ``punctured'' on a large set $S\subseteq \F_2^n$. Recall that the binary Reed-Muller code $\mathrm{RM}[d,n]$ consists of codewords in $\F_2^n$ that correspond to the evaluation vectors of degree $\leq d$ polynomials over $\F_2$. Given a set $S\subseteq \F_2^n$, the resulting punctured code consists of the evaluation of degree $\leq d$ polynomials on $S$. In this context, \cref{thm:intro-correlationbound} can be used to bound the list-size of any puncturing of the Reed-Muller code, showing that for any word $w$ from $\F_2^S$, only a small fraction of codewords are within radius $\frac{1}{2}-\eps$ of $w$. Another interpretation of \cref{thm:intro-correlationbound} is that any puncturing of the Reed-Muller codes over a set $S$ can be turned into a ``small-biased'' code without much loss in the rate of the code.


\paragraph{Sampling lower bound for polynomial sources.} Our low-degree extractor for lower-degree sources (\cref{thm:intro-extrForSources}) has a direct application in distributions that are hard to sample by low-degree polynomials. Indeed, an argument similar to the proof of \cite[Lemma 3]{viola2016quadratic}, \cref{thm:intro-extrForSources} implies the existence of a degree-$O(d)$ polynomial $p$ for which the distribution $(\mathbf{U}, p(\mathbf{U}))$ cannot be sampled by any degree-$d$ source, where $\mathbf{U}\sim \mathbf{U}_n$. 

Suppose that $p$ is a degree-$O(d)$ polynomial that is an $\eps$-extractor for the family of degree $\leq 2d$ sources over $\{0,1\}^n$ of min-entropy $\geq \frac{n}{2}$, where $\eps=o(1)$. The existence of such a polynomial $p$ is guaranteed by \cref{thm:intro-extrForSources}. Now suppose that $(G(\mathbf{U'}), g(\mathbf{U'}))$, where $\mathbf{U'}\sim \mathbf{U}_m$ for some $m\geq 1$, is a degree $\leq d$ source sampling $(\mathbf{U}, p(\mathbf{U}))$. In particular, $G$ is an $n$-bit degree $\leq d$ source and $g$ is a degree $\leq d$ polynomial. Consider the $n$-bit random variable $\mathbf{R}=G(\mathbf{U'})\cdot g(\mathbf{U'})+ \mathbf{U}_n\cdot (1-g(\mathbf{U'}))$. Since $\mathbf{R}$ is sampled by a degree $\leq 2d$ source of min-entropy $n-O(1)$, $\Pr[p(\mathbf{R})=1] = \frac{1}{2}+o(1)$. On the other hand, by the definition $\mathbf{R}$, we have $\Pr[p(\mathbf{R})=1] \geq \frac{1}{2}+ \Omega(1)$, which is a contradiction.

\paragraph{Related Work.} 
An independent and concurrent paper by Alrabiah, Goodman, Mosheiff, and Ribeiro~\cite{AGMR24} proves the existence of low-degree extractors for similar families of sources that are considered in our work, as well as sumset sources. While the proofs are quite different, they both rely on bounds on the dimension of punctured Reed-Muller codes (equivalently the Hilbert function).

\paragraph{Acknowledgments.}
We thank Omar Alrabiah, Jesse Goodman, Jonathan Mosheiff, and Jo\~{a}o Ribeiro for sharing with us an early draft of their work. We would also like to thank Jesse Goodman and S. Venkitesh for helpful discussions and pointers. We are very grateful to the anonymous reviewers for their comments and pointers to related work. Part of this work was conducted while the second author was visiting the Simons Institute for the Theory of Computing at UC Berkeley; he extends his thanks to the institute for its support and hospitality.

\section{Preliminaries}
All logarithms in this paper are base 2. By~$\N$ we denote the set of non-negative integers. For a positive integer~$n$, by $[n]$ we denote the set $\{1,\ldots,n\}$.
For a prime power $q$, denote by $\F_q$ the finite field $q$ elements. 

For simplicity, throughout this paper, we refer to a polynomial as a degree-$d$ polynomial if its total degree is \emph{at most} $d$. 
When $q$ is a prime power, by $\Poly{q}{n,d}$ we denote the set of all degree-$d$ polynomials from $\F[X_1,\dots,X_n]$ with individual degrees at most $q-1$. Note that each element of $\Poly{q}{n,d}$ corresponds to a unique map $\F_q^n \rightarrow \F_q$.


Let $r_1,\dots,r_n\geq 1$ be integers and $F=\prod_{i=1}^n\{0,\ldots,r_i-1\}$. For $x \in F$ and $i \in [n]$, $x_i$ denotes the $i$th coordinate of $x$. For $x\in F$, we define its \emph{generalized Hamming weight} as $|x|\coloneqq \sum_i x_i$, where the summation is over the integers. For an integer $d\geq 0$, and a set $T\subseteq  F$, we denote the set of its elements of generalized Hamming weight $\leq d$ by
\[
T_{\leq d}\coloneqq \{x\in T : |x|\leq d\} \;.
\] 

For $a,b \in F$,
we write $a\leq_P b$ if $a_i\leq b_i$ for all $i\in [n]$.
We say a subset $T\subseteq F$ is \emph{down-closed} if for all $a,b\in F$ such that $a\leq_P b$, if $b$ is in $T$, then so is $a$.
Similarly, we say a subset $T\subseteq F$ is \emph{up-closed} if for all $a,b\in F$ such that $a\leq_P b$, if $a$ is in $T$, then so is $b$.

The \emph{lexicographic order} $\prec$ on $F$ is defined as follows. For distinct $x, y \in F$, $x$ precedes $y$, denoted $x\prec y$, in lexicographic order if $x_i<y_i$, where $i$ is the smallest index such that $x_i\neq y_i$.

We will be studying the following quantity.
\begin{definition}\label{def:h}

For $F=\prod_{i=1}^n\{0,\ldots,r_i-1\}$ and $k\leq |F|$, let 
\[
\mathcal{H}_F(d,k)\coloneqq \min_{T} |T_{\leq d}| \;,
\]
where the minimum is taken over all down-closed sets $T\subseteq F$ with $|T|=k$.  
Moreover, denote $\mathcal{H}_F(d,k)$ by $\mathcal{H}^n_q(d,k)$ in the special case where $r_1=\dots=r_n=q$ for some $q\geq 1$.
\end{definition}

\subsection{Probability Distributions}
We use lowercase letters such as $x,y$ to denote vectors, uppercase bold letters such as $\textbf{X}, \textbf{Y}$ to denote random variables, and $\mathcal{X}, \mathcal{Y}$ to denote families of distributions. By $\textbf{U}_n$ we denote the uniform distribution over $\{0,1\}^n$. 





The statistical distance between two distributions $\mathbf{A}$ and $\mathbf{B}$ over a finite domain $X$ is 
\[\Delta(\mathbf{A}, \mathbf{B})= \frac{1}{2}\left(\sum_{x \in X}  \left| \Pr[x \in \mathbf{A}]- \Pr[x \in \mathbf{B}]\right| \right)\;.\]
We say two distributions $\mathbf{A}$ and $\mathbf{B}$ are $\varepsilon$-close if $\Delta(\mathbf{A}, \mathbf{B}) \leq \varepsilon$. For a distribution $\mathbf{X} \sim \{0,1\}^n$, the min-entropy of~$\mathbf{X}$ is $$\ent(\mathbf{X})=\min_{x \in \mathrm{support}(\mathbf{X})}-\log (\Pr[\mathbf{X}=x])\;.$$

We will use following forms of Chernoff's and Hoeffding's bounds (see, e.g.,~\cite{MU17,hoeffding1963probability}).

\begin{theorem}[Chernoff bound]\label{chernoff}
    Let $X_1,\dots,X_n\in\{0,1\}$ be independent random variables. Let $X=\sum_{i=1}^nX_i$ and $\mu=\mathbb{E}(X)$. Then we have
    $$
    \Pr[|X-\mu|\geq\delta\mu]\leq2e^{-\mu\delta^2/3}
    $$
    for all $0<\delta<1$.
\end{theorem}

\begin{theorem}[Hoeffding's inequality]\label{hoeffding}
Let $X_1,\dots,X_n\in[0,1]$ be independent random variables, $X=\sum_{i=1}^n X_i$ and $\mu=\E[X]$. Then, 
\[
\Pr[|X-\mu|\ge R] \leq 2e^{-\frac{2R^2}{n}}. 
\]
\end{theorem}

\subsection{Randomness Sources, Dispersers, and Extractors}

\begin{definition}[Sources and Their Convex Combinations]
    A distribution $\mathbf{X} \sim \{0,1\}^n$ is a \emph{source} from a class $\mathcal{C}$ of functions, if $\mathbf{X}= f(\mathbf{U}_m)$ for some $f: \mathcal\{0,1\}^m \rightarrow \{0,1\}^n \in \mathcal{C}$. A distribution $\mathbf{Y}$ is a \emph{convex combination} of sources $\mathbf{X}_i$ if  $\mathbf{Y}=\sum_i p_i \mathbf{X}_i$ for some non-negative $p_i$ satisfying $\sum_i p_i=1$, i.e., $\mathbf{Y}$ samples from each $\mathbf{X}_i$ with probability $p_i$.
\end{definition}






One of the most powerful classes of sources that we consider in this work is the class of circuits of polynomial size. 
\begin{definition}[$\ACp$ circuits]
    An \emph{$\ACp$ circuit} is an unbounded-depth Boolean circuit consisting of $\AND$, $\OR$, $\XOR$, $\NOT$ gates of unbounded fan-in. The size of such a circuit is the number of non-input gates in it.
\end{definition}
We focus on the class of $\ACp$ circuit as it generalizes circuit classes previously studied in this context: unbounded-depth circuits of bounded fan-in from $\ppoly$, and bounded-depth circuits of unbounded fan-in from, say, $\AC^0$. We remark that we define $\ACp$ sources (see \cref{def:sources}) as sources where each output is computed by an $\ACp$ circuit of polynomial size but with an arbitrary (possibly super-polynomial) number of inputs. This explains  why in this context $\ppoly$ and $\AC^0$ circuits are incomparable, and why we work with $\ACp$ circuits generalizing both of the aforementioned classes. In fact, our results hold even for a larger class of circuits where not only $\XOR$ but arbitrary constant-degree polynomials over $\F_2$ can be computed at gates (see the discussion at the end of \cref{sec:dispersers}).

\begin{definition}[Structured Sources]\label{def:sources}
Let $n,d,m\in\N$, $f:\{0,1\}^m\rightarrow \{0,1\}^n$, and $\mathbf{X}$ be a distribution over $\{0,1\}^n$ that is generated as $f(\mathbf{U}_m)$.
\begin{itemize}
\item $\mathbf{X}$ is called a \emph{$d$-local source} if every output bit of $f$ depends only on at most $d$ of its input~bits. 
\item $\mathbf{X}$ is called a \emph{depth-$d$ decision forest source} if every output bit of $f$ is determined by a depth-$d$ decision tree of its input variables. 
\item $\mathbf{X}$ is called a \emph{degree-$d$ source} if every output bit of $f$ is a degree-$d$ polynomial over $\F_2$. 
\item $\mathbf{X}$ is called a \emph{size-$n^d$ circuit source} if there is an $\ACp$ circuit of size $n^d$ that computes all output bits of $f$. 
\end{itemize} 
\end{definition}

Note that every $d$-local source is a depth-$d$ decision forest source, and a degree-$d$ source. Also, every depth-$d$ decision forest source is a degree-$d$ source and a $2^d$-local source.

We will use the following bounds on the numbers of $d$-local sources and depth-$d$ decision forest sources.
\begin{proposition}\label{prop:countingSources}
Let $n,d\geq1$. 
\begin{itemize}
    \item The number of $d$-local sources over $\{0,1\}^n$ is bounded from above by $2^{2^d n + 2d n\log n}$. 
    \item The number of depth-$d$ decision forest sources is bounded from above by $2^{(d+\log n)2^{d+1} n}$.
\end{itemize}
\end{proposition}
\begin{proof}
    Every $d$-local source over $\{0,1\}^n$ can be expressed as $f(\mathbf{U}_m)$ where $f:\{0,1\}^m \to \{0,1\}^n$ is a $d$-local function, and we may assume without loss of generality that $m\leq d n$, as the number of relevant variables is bounded from above by $d n$. Then the number of distinct $d$-local sources is at most 
\[
\left(\binom{dn}{d}\cdot 2^{2^d}\right)^n\leq 
2^{2^d n + 2d n\log n}\;.
\]

Similarly, for depth-$d$ decision forest sources we may assume that $m\leq 2^d n$, as each depth-$d$ decision tree depends on at most $2^d$ variables. The number of depth-$d$ decision trees on $m$ variables is at most $(2m)^{2^d}$. Thus, the number of distinct depth-$d$ decision forest sources on $\{0,1\}^n$ is at most 
\[
2^{(d+\log n)2^{d+1} n}\;. \qedhere
\]
\end{proof}
For polynomial and circuit sources where the number of input bits cannot be bounded by a small function of~$n$ (unlike the sources considered in \cref{prop:countingSources}), we will need the following bounds on the number of such sources for a \emph{fixed} number of input bits~$m$.

\begin{proposition}\label{prop:countPolys}
    Let $n,d,m\geq1$. 
    \begin{itemize}
        \item The number of degree-$d$ polynomials $f\colon\F_2^m\to\F_2^n$ is bounded from above by $2^{n\cdot\binom{m}{\leq d}}$. 
        \item The number of $\ACp$ circuits $C\colon\{0,1\}^m\to\{0,1\}^n$ of size $n^d$ is bounded from above by $2^{4n^d(n^d+m)}$.
    \end{itemize}
\end{proposition}
\begin{proof}
The bound on the number of degree-$d$ polynomials follows from the observation that the number of multilinear monomials in $m$ variables of degree at most $d$ is $\binom{m}{\leq d}$. 

For the bound on the number of $n^d$-size circuits, first notice that each gate of~$C$ can be described by $m+n^d$ bits specifying the set of gates feeding it, and $3$ bits specifying the function computed at the gate. Additional $n\log(n^d+m)$ bits are sufficient to specify the $n$ output bits of the circuit. Therefore, the total number of such circuits is bounded from above by
\[
2^{n^d(m+n^d+3)+n\log(n^d+m)} \leq 2^{4n^d(n^d+m)}\;.\qedhere
\]
\end{proof}

\begin{definition}[Disperser]
    A function $\mathrm{Disp}:\{0,1\}^n \rightarrow \{0,1\}$ is a disperser for a family $\mathcal{X}$ of sources over $\{0,1\}^n$ with min-entropy $k$, if for every source $\textbf{X}\in \mathcal{X}$ with $\ent(\textbf{X})\geq k$, the support of $\mathrm{Disp}(\textbf{X})$ is $\{0,1\}$. 
\end{definition}

\begin{definition}[Extractor]
    A function $\mathrm{Ext}:\{0,1\}^n \rightarrow \{0,1\}^m$ is an $\varepsilon$-extractor for a family $\mathcal{X}$ of sources over $\{0,1\}^n$ with min-entropy $k$, if for every source $\mathbf{X} \in \mathcal{X}$ with $\ent(\mathbf{X})\geq k$,  $\Delta\left(\mathrm{Ext}(\mathbf{X}(\mathbf{U}_t)),\mathbf{U}_m\right) \leq \varepsilon$.
\end{definition}

For clarity of presentation, in this paper when working with sources that are guaranteed to have entropy $\ent(\mathbf{X})\geq k$, we will always assume that $k$ is an integer.

\subsection{Hilbert Functions and Standard Monomials}
In this section, we recall some necessary definitions (see, e.g.,~\cite{CLO13}). 
Let $\F$ be a field, $X_1,\ldots,X_n$ be indeterminates, and $\F[X_1,\dots,X_n]$ be the polynomial ring in~$n$ indeterminates over~$\F$. For a polynomial $f\in \mathbb F[X_1,\dots,X_n]$ and $S\subseteq\mathbb\F^n$, let $f|_S\in \F^{S}$ be the restriction of~$f$ to~$S$. For $d\in\N$, by $\Gamma_S(d)\subseteq \F^{S}$ we denote the vector space spanned by $f|_S$ for all degree-$d$ polynomials~$f$:
\[
\h_U(d, \F_q):=\{f|_S\colon f\in \F[X_1,\dots,X_n], \deg(f)\leq d\} \;.
\]
\begin{definition}[Hilbert function]
For a set $S\subseteq \mathbb F^n$, the \emph{(affine) Hilbert function} of $S$ over $\F$, $\h_S(\mathrel{\;\cdot\;}, \F)\colon\N\to\N$, is defined as the dimension of $\Gamma_S(d)$ over $\F$, i.e.,
\[
\h_S(d, \F):=\dim_{\F}\left(\Gamma_S(d)\right)\;.
\]
\end{definition}

\begin{definition}[Monomial order]
Let $\preceq$ be a total order on the monomials in a polynomial ring $\F[X_1,\ldots,X_n]$. The order $\preceq$ is called a \emph{monomial order} if $1$ is the minimal element of~$\preceq$, and for all monomials $m_1, m_2, m$ satisfying $m_1\preceq m_2$, we have that $m_1 m\preceq m_2 m$. The order $\preceq$ is \emph{degree-compatible} if for all monomials $m_1, m_2$ such that $\deg(m_1)<\deg(m_2)$, we have that $m_1\preceq m_2$. 
\end{definition}
Examples of degree-compatible monomial orders include the graded lexicographic and graded reverse lexicographic orders.

\begin{definition}[Graded orders]
    The \emph{graded lexicographic order} $\le_{\mathrm{grlex}}$ and the \emph{graded reverse lexicographic order} $\le_{\mathrm{grevlex}}$ are defined as follows. For a pair of monomials $m_1=X_1^{\alpha_1}\cdots X_n^{\alpha_n}$ and $m_2=X_1^{\beta_1}\cdots X_n^{\beta_n}$, let $\alpha=\sum_{i=1}^n \alpha_i$, $\beta=\sum_{i=1}^n \beta_i$, and $\gamma=(\beta_1-\alpha_1,\ldots,\beta_n-\alpha_n)$. We have that $m_1 \le_{\mathrm{grlex}} m_2$ if and only if either $\alpha < \beta$, or $\alpha=\beta$ and the leftmost non-zero entry of $\gamma$ is positive. Similarly, $m_1 \le_{\mathrm{grevlex}} m_2$ if and only if either $\alpha < \beta$, or $\alpha=\beta$ and the rightmost non-zero entry of $\gamma$ is negative.
\end{definition}

\begin{definition}[Leading monomial]
For a nonzero polynomial $f\in\F[X_1,\dots,X_n]$, the \emph{leading monomial} of $f$ under a monomial order $\preceq$ is the largest monomial of $f$ under $\preceq$.
\end{definition}

Let $R$ be a commutative ring (such as the polynomial ring $\F[X_1,\dots,X_n]$). An \emph{ideal} of $R$ is a subset $I$ of $R$ such that for all $a,b\in I$ and $r\in R$, we have that $a+b\in I$ and $ra\in I$.

\begin{definition}[Standard monomial]
Let $I$ be an ideal of $\F[X_1,\dots,X_n]$, and $\preceq$ be a monomial order. A \emph{standard monomial}~$m$ of $I$ is a monomial in $X_1,\dots,X_n$ that is not the leading monomial of any nonzero polynomial in~$I$. 
\end{definition}

For an ideal $I$ and $d\in\N$, $\SM(I)$ denotes the set of all standard monomials of $I$, and $\SM_{\leq d}(I)$ denotes the set of all standard monomials of $I$ of degree at most~$d$:
\[\SM_{\leq d}(I) = \{m\in \SM(I)\colon \deg(m)\leq d\} \;.\]
For a set $S\subseteq\F^n$, by $I(S)$ we denote the ideal of polynomials in $\F[X_1,\dots,X_n]$ vanishing on~$S$,
\[I(S)=\{f\in\mathbb F[X_1,\dots,X_n]: f|_S=0^{S}\}\,.\]
For an ideal $I$ of $\F[X_1,\dots,X_n]$, define the set $V(I)\subseteq\F^n$ by
\[
V(I)=\{a\in\F^n: f(a)=0~\text{for all}~f\in I\} \;.
\]
By definition, for all $f\in I$ and $a\in V(I)$, we have $f(a)=0$. So $I\subseteq I(V(I))$.

Finally, for a set $S\subseteq\F^n$, define
\begin{align*}
    \SM(S) = \SM(I(S))
    \quad \text{and} \quad
    \SM_{\leq d}(S) = \SM_{\leq d}(I(S))\;.
\end{align*}

We say that a set $T$ of monomials is \emph{down-closed} if for all monomials $m$ and $m'$ such that $m\in T$ and $m'$ divides $m$, it holds that $m'\in T$.
It is easy to see that $\mathrm{SM}(S)$ is down-closed. Indeed, if $m'$ was the leading monomial of a polynomial $p\in I(S)$, then $m$ would be the leading monomial of the polynomial $p\cdot(m / m') \in I(S)$.

We will use the following facts about $\SM(S)$ and $\SM_{\leq d}(S)$, which are proven, for example, in \cite[Lemma~1]{R16} and \cite[Corollary 2.1.21]{F07}.
\begin{lemma}\label{lem:standard-cardinality}
Let $S\subseteq\mathbb{F}^n$ be a finite set. Then
\begin{enumerate}
    \item[(a)] for every monomial order $\preceq$, 
    \[|S|=|\mathrm{SM}(S)| \;;\]
    \item[(b)] for every degree-compatible monomial order $\preceq$ and every $d\in\N$, 
    \[\h_S(d, \F)=|\mathrm{SM}_{\leq d}(S)|\;.\]
\end{enumerate} 
\end{lemma}

\section{Hilbert Functions of Sets in Finite Grids}

Let $\F$ be a field.
We consider Hilbert functions of subsets of a finite grid $A=\prod_{i=1}^n A_i$, where each~$A_i$ is a finite subset of the field $\F$.
The main result of this section is that the minimum value $\h_S(d,\F)$ of a set $S\subseteq A$ of size $k$ equals the quantity $\mathcal{H}_F(d,k)$ introduced in \cref{def:h}, where $F=\prod_{i=1}^n \{0,1,\dots, |A_i|-1\}$.

Consider the following setting: 
Let $r_1,\ldots, r_n$ be integers such that $1\leq r_i\leq |\F|$ for $i\in [n]$. 
For each $i\in [n]$, let $A_i$ be a subset of $\F$ consisting of $r_i$ distinct elements $a_{i,1},\dots,a_{i,r_i}\in\F$. 
Let $A$ be the Cartesian product $\prod_{i=1}^n A_i$.
Let $\mathcal{M}$ be the set of monomials dividing $\prod_{i=1}^n X_i^{r_i-1}$. 
Let $\sigma_A$ be the bijection from $\mathcal{M}$ to $A$ defined by 
\begin{equation}\label{eq:sigma-A}
\sigma_A: \prod_{i=1}^n X_i^{e_i} \mapsto (a_{1, e_1+1},\dots, a_{n, e_n+1}).
\end{equation}
Finally, fix a degree-compatible monomial order $\preceq$.

The next lemma states that every down-closed subset $T\subseteq \mathcal{M}$ can be realized as the set of standard monomials of the set $\sigma_A(T)\subseteq A$.

\begin{lemma}\label{lem:realizable}
Let $T$ be a down-closed subset of $\mathcal{M}$.
Then $\mathrm{SM}(\sigma_A(T))=T$.
\end{lemma}

\begin{proof}
Let $I$ be the ideal of $\F[X_1,\dots,X_n]$ generated by the set of polynomials 
\[
\left\{
\prod_{i=1}^{n}\prod_{j=1}^{\deg_{X_i}(m)} (X_i-a_{i,j}): m\in \mathcal{M}\setminus T
\right\}\cup\left\{\prod_{j=1}^{r_1} (X_1-a_{1,j}),\dots, \prod_{j=1}^{r_n} (X_n-a_{n,j})\right\}\;.
\]
Let $S=V(I)$, so that $I\subseteq I(S)$. As $\prod_{j=1}^{r_1} (X_1-a_{1,j}),\dots, \prod_{j=1}^{r_n} (X_n-a_{n,j})\in I$, we have 
\[
S\subseteq V\left(\prod_{j=1}^{r_1} (X_1-a_{1,j}),\dots, \prod_{j=1}^{r_n} (X_n-a_{n,j})\right)=A\;.
\]
Next, we show that $\SM(S)\subseteq T$. Consider $m\in \SM(S)$. Suppose $\deg_{X_i}(m)\geq r_i$ for some $i\in [n]$. Then $m$ is the leading monomial of $(m/X_i^{r_i})\cdot\prod_{j=1}^{r_i} (X_i-a_{i,j})$, and the latter is in $I(A)\subseteq I(S)$. This contradicts the assumption that $m\in \SM(S)=\SM(I(S))$.
Therefore, in the following we assume that for all $i\in [n]$, $\deg_{X_i}(m)<r_i$ , or equivalently, $m\in\mathcal{M}$. Now suppose $m\not\in T$. Then by the definition of $I$, the polynomial $\prod_{i=1}^{n}\prod_{j=1}^{\deg_{X_i}(m)} (X_i-a_{i,j})$ is in $I\subseteq I(S)$, and its leading monomial is $m$. This again contradicts the assumption that $m\in \SM(S)=\SM(I(S))$. This shows that every $m\in \SM(S)$ also belongs to~$T$. Therefore, $\SM(S)\subseteq T$.

Consider arbitrary $m_0=\prod_{i=1}^n X_i^{e_i}\in T$ and let $a=\sigma_A(m_0)=(a_{1, e_1+1},\dots, a_{n, e_n+1})$. Consider arbitrary $m\in \mathcal{M}\setminus T$. As $T$ is down-closed, there exists $k\in [n]$ such that $\deg_{X_{k}}(m)>\deg_{X_{k}}(m_0)=e_k$.
The polynomial $\prod_{i=1}^{n}\prod_{j=1}^{\deg_{X_i}(m)} (X_i-a_{i,j})$ then contains the factor $X_k-a_{k, e_k+1}$, so it vanishes at $a$. As $m\in \mathcal{M}\setminus T$ is arbitrary and $\prod_{j=1}^{r_1} (X_1-a_{1,j}),\dots, \prod_{j=1}^{r_n} (X_n-a_{n,j})$ also vanish at $a$, this shows $a\in V(I)=S$ by the definition of $I$. As $m_0\in T$ is arbitrary, we see that $\sigma_A(T)\subseteq S$. It follows that 
\[
|T|=|\sigma_A(T)|\leq |S|=|\mathrm{SM}(S)|
\]
where the first equality follows from the injectivity of $\sigma_A$ and the last equality holds by \cref{lem:standard-cardinality}\,(a). As $\SM(S)\subseteq T$, we must have $\mathrm{SM}(S)=T$ and $|\sigma_A(T)|=|S|$. The latter equality together with $\sigma_A(T)\subseteq S$ yields $\sigma_A(T)=S$.
So $\mathrm{SM}(\sigma_A(T))=\mathrm{SM}(S)=T$.
\end{proof}

\begin{lemma}\label{lem:minimum-of-hilbert}
Let $k,d\in\mathbb{N}$ such that $k\leq |A|$. Then
\[
\min_{S\subseteq A: |S|=k} \h_S(d, \F)=\min_{\textup{down-closed}~T\subseteq \mathcal{M}: |T|=k} |\{m\in T: \deg(m)\leq d\}|\;.
\]
\end{lemma}

\begin{proof}
Let $S\subseteq A$ be a set of size $k$ such that $\h_S(d,\F)$ achieves the minimum. Let $T=\SM(S)$, which is down-closed. 
Moreover, as $S\subseteq A$, we have $T=\SM(S)\subseteq \SM(A)=\mathcal{M}$.
Note that $|T|=|\SM(S)|=|S|=k$ by \cref{lem:standard-cardinality}\,(a). And $|\{m\in T: \deg(m)\leq d\}|=|\SM_{\leq d}(S)|=\h_S(d,\F)$ by \cref{lem:standard-cardinality}\,(b). So LHS $\geq$ RHS.

Conversely, let $T$ be a down-closed subset of $\mathcal{M}$ of size $k$ such that $|\{m\in T: \deg(m)\leq d\}|$ achieves the minimum. 
Let $S=\sigma_A(T)$.
Then $\mathrm{SM}(S)=T$ by \cref{lem:realizable}. 
We have $|S|=|\SM(S)|=|T|=k$ by \cref{lem:standard-cardinality}\,(a) and $\h_S(d,\F)=|\SM_{\leq d}(S)|=|\{m\in T: \deg(m)\leq d\}|$ by \cref{lem:standard-cardinality}\,(b). So  LHS $\leq$ RHS. This finishes the proof of the lemma.
\end{proof}



Let $F=\prod_{i=1}^n \{0,1,\dots,r_i-1\}$.
Let $\phi: \mathcal{M}\to F$ be the bijection
\begin{equation}\label{eq:phi}
\phi: \prod_{i=1}^n X_i^{e_i}\mapsto (e_1,\dots,e_n).
\end{equation}
\cref{lem:minimum-of-hilbert} can now be reformulated as follows.

\begin{corollary}\label{cor:characterization}
Let $k,d\in\mathbb{N}$ such that $k\leq |A|$. Then
\begin{align}\label{eq:characterization}
\min_{S\subseteq A: |S|=k} \h_S(d,\F)=\mathcal{H}_F(d,k)\;.
\end{align}
\end{corollary}
\begin{proof}
Recall that by \cref{def:h},
\[
\mathcal{H}_F(d,k)=\min_{\textup{down-closed}~S\subseteq F: |S|=k} |S_{\leq d}|\;.
\]
Identify $\mathcal{M}$ with $F$ via the bijection $\phi$. The claim follows from \cref{lem:minimum-of-hilbert} by noting that $T\subseteq \mathcal{M}$ is down-closed iff $\phi(T)\subseteq F$ is down-closed, and that for $m\in \mathcal{M}$, $\deg(m)\leq d$ iff $\phi(m)\in F_{\leq d}$.
\end{proof}

For the special case of a finite field $\F=\F_q$, and $r_1=\dots=r_n=q-1$, we have $A=\F_q^n$, and the right-hand side of \cref{eq:characterization} becomes $\mathcal{H}_q^n(d,k)$ from \cref{def:h}. This leads us to the following corollary.
\begin{corollary}\label{cor:hH}
    For every $n,k,d\in \N$ where $k\leq q^n$, a prime power $q$, and every set $S\subseteq \F_q^n$ of size $|S|=k$, we have that
    \[
    \h_S(d,\F_q)\geq \mathcal{H}_q^n(d,k)\;.
    \]
\end{corollary}

Finally, we state the following lemma, which will be used in \cref{sec:closure}. Its proof reuses ideas from the previous proofs in this section.

\begin{lemma}\label{lem:A-to-F}
Let $n,d\in \N$.
Let $\sigma_A: \mathcal{M}\to A$ and $\phi: \mathcal{M}\to F$ be the bijections \eqref{eq:sigma-A} and \eqref{eq:phi} respectively.
Let $S\subseteq A$ such that $T:=\sigma_A^{-1}(S)\subseteq\mathcal{M}$ is down-closed.
Let $T'=\phi(T)\subseteq F$.
Then $\h_S(d,\F)=T'_{\leq d}$.
\end{lemma}

\begin{proof}
By definition, we have $\sigma_A(T)=S$.
So $\mathrm{SM}(S)=\mathrm{SM}(\sigma_A(T))=T$ by \cref{lem:realizable}. 
Then we have
\[
\h_S(d,\F)=|\SM_{\leq d}(S)|=|\{m\in \SM(S): \deg(m)\leq d\}|=|\{m\in T: \deg(m)\leq d\}|=|T'_{\leq d}|.
\]
where the first inequality holds by \cref{lem:standard-cardinality}\,(b), the third one holds by the fact that $\SM(S)=T$, and the last one holds since $\phi$ maps $\{m\in T: \deg(m)\leq d\}$ bijectively to $T'_{\leq d}$.
\end{proof}





\section{Number of Points with Low Hamming Weight in Down-Closed Sets}
In this section, we will find the exact values of all $\mathcal{H}_q^n(d,k)$ which, by \cref{cor:hH}, will give us tight lower bounds on the Hilbert function of sets of size~$k$. 

For every $n,k,q$ where $k\leq q^n$, we define $\mathrm{M}^n_q(k)$ as the set of the first $k$ elements of $\{0,\dots,q-1\}^n$ in lexicographic order.

The main result of this section is the following theorem.
\begin{theorem}\label{thm:hilbertbound}
For every $n,k,d, q\in \N$ where $k\leq q^n$,
\[
\mathcal{H}^n_q(d,k)= |\mathrm{M}^n_q(k)_{\leq d}|\;.
\]
\end{theorem}
Combining \cref{cor:hH} and \cref{thm:hilbertbound}, we obtain the following bounds on the Hilbert function.
\begin{corollary}\label{cor:boundsOnHilbert}
For every prime power $q$, and $n,k,d\in \N$ where $k\leq q^n$, we have 
    \[
    \min_{S\subseteq \F_q^n: |S|=k} \h_S(d,\F_q)= |\mathrm{M}^n_q(k)_{\leq d}|\;.
    \]
In particular, setting $q=2$, for every $n,k,d\in\N$ where $k\leq2^n$, and every $S\subseteq\F_2^n$ of size $|S|=k$,
\[
    \h_S(d,\F_2)\geq \binom{\floor{\log(k)}}{\leq d}\;.
\]
\end{corollary}


We will use the following notation: 
For $t\in\{0,1,\dots,n\}$, define $\mathcal{D}_q^n(t)$ to be the set of $x\in \{0,\dots,q-1\}^n$ whose first $n-t$ coordinates are zero.

Note that for every $q,k$, and $n$,
\[\mathcal{D}_q^n\left(\lfloor \log_q k\rfloor\right) \subseteq \mathrm{M}^n_q(k) \subseteq \mathcal{D}_q^n\left(\lceil \log_q k\rceil\right)\;.
\]
When $n$ and $q$ are clear from the context, we omit the superscript $n$ and the subscript $q$ from $\mathrm{M}^n_q(k)$, $\mathcal{D}^n_q(t)$, and $\mathcal{H}^n_q(d,k)$. 

\subsection{The Boolean Case, \texorpdfstring{$q=2$}{q=2}} 


For a set $S\subseteq\{0,1\}^n$, let $\min(S)$ and $\max(S)$ be respectively the smallest and the largest strings in $S$ in lexicographic order. We say a set $S\subseteq \{0,1\}^n$ is a \emph{contiguous $k$-set} if $|S|=k$ and $S$ consists of all $x$ such that $\min(S)\preceq x \preceq \max(S)$. 

We first show that $\mathrm{M}(k)$ has the largest number of low Hamming weight strings among all contiguous $k$-sets.

\begin{lemma}\label{Lem:number of small HW}
    Let $n,k,d \in \N$ be integers such that $k \leq 2^n$. Let $S^k \subseteq \{0,1\}^n$ be a contiguous $k$-set. Then $|\mathrm{M}(k)_{\leq d}|\geq |S^k_{\leq d}|$.
\end{lemma}
\begin{proof}
We prove this lemma by induction on $k$ and $\max(S^k)$. The base cases of $k\leq 1$ hold trivially, as in this cases $|\mathrm{M}(k)_{\leq d}|=k$. For each $k$, the base case of minimal $\max(S^k)$ holds trivially as in this case $S^k=\mathrm{M}(k)$ and $|S^k_{\leq d}|=|\mathrm{M}(k)_{\leq d}|$.

Assume for the induction hypothesis that the above statement is true for smaller contiguous sets (i.e. $|\mathrm{M}(t)_{\leq d}| \geq |S^{t}_{\leq d}|$ for all contiguous sets $S^{t}$ of size $t \leq k-1$) as well as for contiguous sets $A^{k}$ of size $k$ having $\max(A^k) \prec \max(S^k)$.

Let $m$ be the smallest integer such that $S^k \subseteq \mathcal{D}(m)$. If $d \geq m$, then $|S^{k}_{\leq d}|= k$. Also, since $S^{k}$ is a contiguous set such that $S^k \subseteq \mathcal{D}(m)$, we have $\mathrm{M}(k) \subseteq \mathcal{D}(m)$, implying $|\mathrm{M}(k)_{\leq d}|=|\mathrm{M}(k)|=k$. Thus for the remainder of the proof we assume $d<m$. 

If $T=S^k\cap \mathrm{M}(k)\neq \emptyset$, then the claim follows by applying the induction hypothesis to the contiguous $(k-|T|)$-set $S^k-T$. Thus, we may assume from now that $S^k\cap \mathrm{M}(k)=\emptyset$. 

\paragraph{Case 1:} Assume $x_{n-m+1}=1$ for all $x \in S^k$. Define the contiguous $k$-set 
\[
B^k= \{x - e_{n-m+1}: x \in S^k\}\;,
\]
where $e_i$ denotes the element in $\{0,1\}^n$ whose $i$th coordinate is one, and all other coordinates are zeros. Since $\max(B^k) \prec \max(S^k)$, by the induction hypothesis, $|\mathrm{M}(k)_{\leq d}| \geq |B^k_{\leq d}|$. Also, $|B^k_{\leq d}| \geq |S^k_{\leq d}|$ by definition, thereby implying $|\mathrm{M}(k)_{\leq d}| \geq |S^k_{\leq d}|$.

\paragraph{Case 2: } Let $k^*>0$ be the number of elements $x$ in $S^k$ with $x_{n-m+1}=0$. Define 
\[
S^{k^*}= \{x \in S^k: x_{n-m+1}=0\}
\quad\text{and}\quad S^{k-k^*}= S^{k}\backslash S^{k^*}\;.
\]
Note that in this case, we are guaranteed that $S^{k-k*}$ and $S^{k^*}$ are both non-empty contiguous sets. 

Since $k-k^*<k$ and $S^{k-k^*}$ is a contiguous $(k-k^*)$-set, the induction hypothesis gives us 
\[
|\mathrm{M}(k-k^*)_{\leq d}| \geq |S^{k-k^*}_{\leq d}|\;.
\]

It suffices to show that $|T_{\leq d}| \geq |S^{k^*}_{\leq d}|$ where $T= \mathrm{M}(k)\backslash \mathrm{M}(k-k^*)$.  For $U \subseteq \{0,1\}^n$ and $r \leq n$, define the set $\gamma(U,r)=\{0^{n-r}1^r-x: x \in U\}$, where $0^{n-r}1^r-x$ denotes coordinate-wise subtraction. Since $\max(S^{k^*})=0^{n-m+1}1^{m-1}$, $S^{k^*}$ consists of the $k^*$ lexicographically largest elements $\preceq 0^{n-m+1}1^{m-1}$. Consequently, $\gamma(S^{k^*},m-1)=\mathrm{M}(k^*)$. Note that $\gamma(T,m-1)$ is a contiguous $k^*$-set, hence by the induction hypothesis, we get 
\[
|\gamma(S^{k^*}, m-1)_{\leq m-1-d}| \geq |\gamma(T, m-1)_{\leq m-1-d}|\;,
\]
where we used the fact that $d\leq m-1$. This implies $|T_{\leq d}| \geq |S^{k^*}_{\leq d}|$, as $|S^{k^*}_{\leq d}|=k^*-|\gamma(S^{k^*}, m-1)_{\leq m-1-d}|$ and $|T_{\leq d}| = k^*- |\gamma(T, m-1)_{\leq m-1-d}|$, where we used the fact that $\max(T)\preceq 0^{n-m+1}1^{m-1}$ as $\mathrm{M}(k)\cap S^k =\emptyset$. 
\end{proof}

We now use \cref{Lem:number of small HW} to prove that if a contiguous $k$-set $S^k$ that does \emph{not} contain any of the first $k$ strings in lexicographic order, then the result of \cref{Lem:number of small HW} $|S^k_{\leq d}|\leq |\mathrm{M}(k)_{\leq {d}}|$ can be strengthened to $|S^k_{\leq d}|\leq |\mathrm{M}(k)_{\leq {d-1}}|$.

\begin{lemma}\label{lem:first k strings-low hamming weight-binary case}
 Let $n,k,d \in \N$ be integers such that $k \leq 2^n$. Let $S^k \subseteq \{0,1\}^n$ be a contiguous $k$-set. If~${S^k\cap \mathrm{M}(k)=\emptyset}$, then $|\mathrm{M}(k)_{\leq {d-1}}|\geq |S^k_{\leq d}|$.
\end{lemma}
\begin{proof}
We prove this lemma by induction on~$k$. The base cases of $k\leq 1$ hold trivially. Suppose the statement is true for contiguous $(\leq k-1)$-sets. Let $S^k$ be a contiguous $k$-set and let $m$ be the smallest integer such that $S^k \subseteq \mathcal{D}(m)$.
    
    \textbf{Case 1:} If $x_{n-m+1}=1$ for all $x \in S^k$, consider $B^k=\{x - e_{n-m+1}: x \in S^k\}$. \cref{Lem:number of small HW} shows $|\mathrm{M}(k)_{\leq d-1}| \geq |B^k_{\leq d-1}|= |S^k_{\leq d}|$.
    
    \textbf{Case 2:}  Define the contiguous $k^*$-set $S^{k^*}= \{x \in S^k: x_{n-m+1}=0\}$. We denote $S^{k-k^*}=S^k \backslash S^{k^*}$ and $T=\mathrm{M}(k) \backslash \mathrm{M}(k-k^*)$. 
    
    \begin{enumerate}
        \item If $d \geq m-1$, the claim is trivial since $|\mathrm{M}(k)_{\leq d-1}|=|\mathrm{M}(k)|=k,$ as $\mathrm{M}(k)\subseteq \mathcal{D}(m-1)$ due to our assumption of $S^k\cap \mathrm{M}(k)=\emptyset$. 
        \item If $d \leq m-2$, using the analysis in Case~$1$, we get $|\mathrm{M}(k-k^*)_{\leq d-1}| \geq |S^{k-k^*}_{\leq d}|$. It now suffices to show that $|T_{\leq d-1}| \geq |S^{k^*}_{ \leq d}|$, where $T= \mathrm{M}(k)\backslash \mathrm{M}(k-k^*)$. Analogous to the argument made in Case 2 in the proof of \cref{Lem:number of small HW}, one can show that $\mathrm{M}(k^*)= \gamma(S^{k^*}, m-1)$. Since $S^{k^*} \cap \left(\mathrm{M}(k)\backslash \mathrm{M}(k^*)\right) = \emptyset$, we also have $\gamma(S^{k^*},m-1) \cap \gamma(\mathrm{M}(k)\backslash \mathrm{M}(k^*),m-1)= \emptyset$ and by the induction hypothesis, we get $|\mathrm{M}(k^*)_{ \leq m-d-1}| \geq |\gamma(T,m-1)_{\leq m-d}|$. 
        
        Now, note that $|S^{k^*}_{\leq d}|= k^*-|\mathrm{M}(k^*)_{ \leq m-d-1}|$, and since $T\subseteq \mathcal{D}(m-1)$, we have $|T_{\leq d-1}|= k^*-|\gamma(T,m-1)_{\leq m-d}|$. Combining these, we get $|T_{\leq d-1}| \geq |S^{k^*}_{ \leq d}|$, as desired. \qedhere
    \end{enumerate} 
\end{proof}

We are finally ready to prove the Boolean case of \cref{thm:hilbertbound}. 
\begin{proof}[Proof of \cref{thm:hilbertbound}, the $q=2$ case]

 Let $S\subseteq \{0,1\}^n$ be a down-closed set of size $k$. We prove this theorem by a simultaneous induction on $k,d\geq 0$. 
 
 For the base cases, we consider pairs $(k,d)$ such that $d=0$ or $k\leq 2^d$. The case of $d=0$ is trivial. For the case where $k\leq 2^d$, a down-closed set $S$ of size $k$ cannot have strings of Hamming weight $> d$, thereby showing $|S_{\leq d}|=k$. Also, by construction, $\mathrm{M}(k)$ is a down-closed set of size $k$, implying ${\mathcal H}(d,k)=|\mathrm{M}(k)_{\leq d}|=k$ in this case.
 
Given $d\geq 1$ and $k> 2^d$, assume that the theorem is true for all $(k',d')$ such that either $k'<k$, or $k'=k$ and $d' <d$. Suppose $S$ is a down-closed set of size $k$ and let $m$ be the smallest integer such that $S \subseteq \mathcal{D}(m)$. Define
     \begin{align*}
        S^0 &\coloneqq \{x\in S : x_{n-m+1}=0\}\;,\\
        S^1 &\coloneqq \{x-e_{n-m+1} : x\in S~\text{and}~x_{n-m+1}=1\}\;.
        \end{align*}
        Since $S$ is down-closed, we have $S^1 \subseteq S^0$. Moreover, 
        \[
        |S_{\leq d}| = |S^0_{\leq d}| + |S^1_{\leq d-1}|\,. 
        \]
          Applying the induction hypothesis for $k'=|S^0|<k$ and $d$, we get $|\mathrm{M}(|S^0|)_{\leq d}| \leq |S^0_{\leq d}|$. Let $T=\mathrm{M}(k)\backslash \mathrm{M}(|S^0|)$. Since $|S^1| \leq |S^0|$, we have $\mathrm{M}(|S^1|) \cap T= \emptyset$, and we may apply  \cref{lem:first k strings-low hamming weight-binary case} to get $ |T_{\leq d}|\leq |\mathrm{M}(|S^1|)_{\leq d-1}|$. Now applying the induction hypothesis for $k'=|S^1|$ and $d'=d-1$, we get $|\mathrm{M}(|S^1|)_{\leq d-1}| \leq |S^1_{\leq d-1}|$.  Combining these observations, we get
\begin{align*}
|\mathrm{M}(k)_{\leq d}|&=|\mathrm{M}(|S^0|)_{\leq d}|+|T_{\leq d}|\\
&\leq |S^0_{\leq d}|+|\mathrm{M}(|S^1|)_{\leq d-1}|
\\ 
&\leq |S^0_{\leq d}|+|S^1_{\leq d-1}|\\ &=|S_{\leq d}|.
\end{align*}
This concludes the induction, and shows that for every $k,d\geq 0$, and down-closed set $S$ of size $k$, $|\mathrm{M}(k)_{\leq d}| \leq |S_{\leq d}|$. 
%
%
%
%
%
%
\end{proof}



\subsection{The General Case of Finite Grids}
We prove \cref{thm:hilbertbound} in this subsection. In fact, we prove the theorem in a more general setting, described as follows.

Let $F=\prod_{i=1}^n \{0,1,\dots,r_i-1\}$ where $r_1\leq r_2\leq \dots\leq r_n$. Let $d\in\mathbb{N}$. We introduce the following notations:

For $S\subseteq F$, define $\nabla(S):=\{a\in F: b\leq_P a~\text{for some}~b\in S\}$, i.e., $\nabla(S)$ is the up-closure of $S$.
For $k\in\{0,\dots,|F|\}$, denote by $\mathrm{M}(k)$ the set of the smallest $k$ elements of $F$ in lexicographic order.
And for $r\in\{0,\dots,|F_{\leq d}|\}$,
denote by $\mathrm{L}_{\leq d}(r)$ the set of the largest $r$ elements of $F_{\leq d}$ in lexicographic order. 

The main result of this subsection is the following generalization of \cref{thm:hilbertbound}.

\begin{theorem}\label{thm:general-BD}
For every $k\in\mathbb{N}$ such that $k\leq |F|$,
\[
\mathcal{H}_F(d,k)=|\mathrm{M}(k)_{\leq d}|\;.
\]
\end{theorem}

We derive \cref{thm:general-BD} from a combinatorial result of Beelen and Datta \cite{BD18}, which generalizes the earlier work of Wei \cite{Wei91} and Heijnen--Pellikaan \cite{HP98, Hei99}.

\begin{theorem}[{\cite[Theorem~3.8]{BD18}}]\label{thm:BD}
Let $S\subseteq F_{\leq d}$ and $r=|S|$. Then $|\nabla(\mathrm{L}_{\leq d}(r))|\leq |\nabla(S)|$.\footnote{In \cite{BD18}, $\mathrm{L}_{\leq d}(r)$ is denoted by $M(r)$, while we use $\mathrm{M}(r)$ to denote the set of the smallest $r$ elements of $F$ in lexicographic order.}
\end{theorem}

Define $\Delta(S):=F\setminus \nabla(S)$ for $S\subseteq F$. 
The next lemma gives a characterization of $\Delta(S)$.

\begin{lemma}\label{lem:delta}
Let $T\subseteq F_{\leq d}$ be down-closed and $S=F_{\leq d}\setminus T$.
Then $\Delta(S)$ is the unique maximal set with respect to inclusion among all down-closed subsets $U$ of $F$ satisfying $U_{\leq d}=T$. 
\end{lemma} 
\begin{proof}
Suppose $\nabla(S)\cap T$ contains an element $a$. Then $b\leq_P a$ for some $b\in S$.
As $T$ is down-closed and $a\in T$, we have $b\in T$ and hence $b\in S\cap T$, contradicting the fact that $S=F_{\leq d}\setminus T$. So $\nabla(S)\cap T=\emptyset$. 
Therefore, $T\subseteq F_{\leq d}\setminus \nabla(S)$.
On the other hand, we have $F_{\leq d}\setminus \nabla(S)\subseteq F_{\leq d}\setminus S\subseteq T$. So $T=F_{\leq d}\setminus \nabla(S)$.

By definition, $\nabla(S)$ is up-closed. So $\Delta(S)=F\setminus \nabla(S)$ is down-closed.
And $\Delta(S)_{\leq d}=F_{\leq d}\setminus\nabla(S)=T$ by definition.

Finally, consider an arbitrary down-closed set $U\subseteq F$ satisfying $U_{\leq d}=T$. Then $S\cap U=\emptyset$.
Applying the first part of the proof to $U$ in place of $T$ then shows that $\nabla(S)\cap U=\emptyset$. So $U\subseteq F\setminus \nabla(S)=\Delta(S)$. This proves the unique maximality of $\Delta(S)$ with respect to inclusion.
\end{proof}

\begin{lemma}\label{lem:DeltaM}
Let $r\in\{0,\dots,|F_{\leq d}|\}$ and $k=|\Delta(\mathrm{L}_{\leq d}(r))|$. Then
$\Delta(\mathrm{L}_{\leq d}(r))=\mathrm{M}(k)$.
\end{lemma}

\begin{proof}
If $r=0$, then $\Delta(\mathrm{L}_{\leq d}(r))=\Delta(\emptyset)=F\setminus\nabla(\emptyset)=F$, and the lemma trivially holds. 
So assume $r>0$.
Let $x=(x_1,\dots,x_n)$ be the smallest element in $\mathrm{L}_{\leq d}(r)$.

Let $U$ be the subset of elements in $F$ smaller than $x$ in lexicographic order.
Then $U\cap \nabla(\mathrm{L}_{\leq d}(r))=\emptyset$ and hence $U\subseteq \Delta(\mathrm{L}_{\leq d}(r))$. 
It suffices to show that $U=\Delta(\mathrm{L}_{\leq d}(r))$.
Consider arbitrary $y=(y_1,\dots,y_n)\in F\setminus U$. Then $y$ is at least $x$ in lexicographic order. 

We claim that $y\in\nabla(\mathrm{L}_{\leq d}(r))$.
If $y=x$, then it is in $\mathrm{L}_{\leq d}(r)\subseteq\nabla(\mathrm{L}_{\leq d}(r))$, so the claim holds in this case.
Now suppose $y$ is greater than $x$ in lexicographic order.
Let $i\in [n]$ be the smallest integer such that $y_i\neq x_i$. We have $y_i>x_i$.
Let $x'=(x_1,\dots,x_{i-1},x_i+1,0,\dots,0)$.

Suppose $x_{i+1},\dots,x_n$ are not all zero. Then the generalized Hamming weight of $x'$ is bounded by that of $x$. So $x'\in \mathrm{L}_{\leq d}(r)$.
As $x'\leq_P y$, we have $y\in \nabla(\mathrm{L}_{\leq d}(r))$.

Now suppose $x_{i+1}=\dots=x_n=0$. Then $x\leq_P y$. Again, this implies $y\in \nabla(\mathrm{L}_{\leq d}(r))$.

So the claim $y\in\nabla(\mathrm{L}_{\leq d}(r))$ holds in all cases. Therefore, $y\not\in\Delta(\mathrm{L}_{\leq d}(r))$. As $y\in F\setminus U$ is arbitrary, this shows $U\supseteq \Delta(\mathrm{L}_{\leq d}(r))$ and hence $U=\Delta(\mathrm{L}_{\leq d}(r))$.
\end{proof}

Now we are ready to prove \cref{thm:general-BD}.


\begin{proof}[Proof of \cref{thm:general-BD}]
For $k\in \{0,1,\dots, |F|\}$, define $f(k)=|\mathrm{M}(k)_{\leq d}|$.
Then 
\begin{equation}\label{eq:monotone}
f(0)=0,~f(|F|)=|F_{\leq d}|,~\text{and}~f(k)\leq f(k+1)\leq f(k)+1~\text{for}~k\in \{0,\dots,|F|-1\}.
\end{equation}

For $s\in\{0,1,\dots,|F_{\leq d}|\}$, let $g(s)$ be the largest integer in $\{0,1,\dots, |F|\}$ such that $f(g(s))=s$.
Note that by \eqref{eq:monotone}, $g(\cdot)$ is a well-defined, strictly monotonically increasing function and $g(|F_{\leq d}|)=|F|$.

Fix $k\leq |F|$. By the definition of $\mathcal{H}_F(d,k)$, we want to show that the smallest possible value of $|T_{\leq d}|$ over all down-closed sets $T\subseteq F$ of size $k$ is $|\mathrm{M}(k)_{\leq d}|$.
Obviously, this value is attained by choosing $T=\mathrm{M}(k)$.

Now assume to the contrary that there exists a down-closed set $T\subseteq F$ of size $k$ such that $|T_{\leq d}|<|\mathrm{M}(k)_{\leq d}|$, i.e., $|T_{\leq d}|<f(k)$.
Let $s$ be the smallest integer in $\{0,1,\dots,|F_{\leq d}|\}$ such that $g(s)\geq k$. Such $s$ exists since $g(|F_{\leq d}|)=|F|$.
Then 
\[
|T_{\leq d}|<f(k)\leq f(g(s))=s
\]
where the second step uses the monotonicity of $f(\cdot)$.
We must have $f(k)>s-1$, since otherwise we would have $s\geq 1$ and $f(k)=s-1$, which implies $g(s-1)\geq k$ by the definition of $g(\cdot)$. But this contradicts the choice of $s$. Therefore, $f(k)=s$.


Let $S=F_{\leq d}\setminus T$.
Let $s'=|T_{\leq d}|<s$ and $r=|S|=|F_{\leq d}|-s'$.
Note that $T_{\leq d}$ is down-closed since $T$ and $F_{\leq d}$ are.
Applying \cref{lem:delta} to $T_{\leq d}$ shows that $\Delta(S)$ is the largest down-closed subset $U$ of $F$ satisfying $U_{\leq d}=T_{\leq d}$. As $T$ is another set satisfying this property, we have 
\begin{equation}\label{eq:Delta-1}
|\Delta(S)|\geq |T|=k \;.
\end{equation}
By \cref{thm:BD},  $|\nabla(\mathrm{L}_{\leq d}(r))|\leq |\nabla(S)|$.
So $|\Delta(\mathrm{L}_{\leq d}(r))|\geq |\Delta(S)|$. Combining this with \eqref{eq:Delta-1} gives 
\[
\Delta(\mathrm{L}_{\leq d}(r))\geq k\;.
\]

Let $k'=|\Delta(\mathrm{L}_{\leq d}(r))|\geq k$. By \cref{lem:DeltaM}, $\Delta(\mathrm{L}_{\leq d}(r))=\mathrm{M}(k')$. Then $f(k')=|\mathrm{M}(k')_{\leq d}|=|\Delta(\mathrm{L}_{\leq d}(r))_{\leq d}|$. And $\Delta(\mathrm{L}_{\leq d}(r))_{\leq d}=F_{\leq d}\setminus \mathrm{L}_{\leq d}(r)$ by \cref{lem:delta}.
So $f(k')=|F_{\leq d}\setminus \mathrm{L}_{\leq d}(r)|=|F_{\leq d}|-r=s'$. 

In summary, we have $s'<s$, $k'\geq k$, $f(k)=s$, and $f(k')=s'$, contradicting the monotonicity of $f(\cdot)$ as stated by \eqref{eq:monotone}.
\end{proof}

\section{A Tight Bound on the Size of \texorpdfstring{Degree-$d$}{Degree-d} Closures of Sets}\label{sec:closure}
For $n,d,\delta\in\N$, denote by $N(n,d,\delta)$ the number of monomials $X_1^{e_1}\cdots X_n^{e_n}$ with $e_1,\dots,e_n\leq \delta$ and $e_1+\dots+e_n\leq d$. For example, $N(n,d,1)=\binom{n}{\leq d}$ and $N(n,d,\delta)=\binom{n+d}{d}$ for $d\leq \delta$.

\begin{lemma}
$\h_{\F_q^n}(d, \F_q)=N(n,d,q-1)$.
\end{lemma}
\begin{proof}
By multivariate Lagrange interpolation, every function on $\F_q^n$ can be uniquely written as a linear combination of $X_1^{e_1}\cdots X_n^{e_n}|_{\F_q^n}$ with $e_1,\dots,e_n\leq q-1$ over $\F_q$. So the set $\{X_1^{e_1}\cdots X_n^{e_n}|_{\F_q^n}: e_1,\dots,e_n\leq q-1, e_1+\dots+e_n\leq d\}$ is a basis of $\Gamma_{\F_q^n}(d)$ of size $N(n,d,q-1)$.
\end{proof}

In particular, \cref{thm:Nie-Wang}, which was proved by Nie and Wang \cite{nie2015hilbert}, can be restated as
\begin{equation}\label{eq:Nie-Wang}
|\cl_d(T)| \leq \frac{q^n}{\h_{\F_q^n}(d, \F_q)} \cdot |T| = \frac{q^n}{N(n,d,q-1)} \cdot |T|\;.
\end{equation}

We now give the following tight bound on the size of the degree-$d$ closure of a set $T\subseteq \F_q^n$, improving \eqref{eq:Nie-Wang}. 

\begin{theorem}\label{thm:tight-bound-cl}
Let $n,d,m\in\N$. Let $T\subseteq\F_q^n$ be a set of size $m$.
Then
\begin{equation}\label{eq:closure-size}
|\cl_d(T)|\leq \max_{0\leq k\leq q^n: |\mathrm{M}^n_q(k)_{\leq d}|\leq m} k = \begin{cases}
\max_{0\leq k\leq q^n: |\mathrm{M}^n_q(k)_{\leq d}|=m} k & \text{if } m\leq N(n,d,q-1),\\
q^n & \text{otherwise.}
\end{cases}
\end{equation}
\end{theorem}

\begin{proof}
Note that $|\mathrm{M}^n_q(q^n)_{\leq d}|=N(n,d,q-1)$ and that $|\mathrm{M}^n_q(k-1)_{\leq d}|\leq |\mathrm{M}^n_q(k)_{\leq d}|\leq |\mathrm{M}^n_q(k-1)_{\leq d}|+1$ for $1\leq k\leq q^n$.
It follows that
\begin{equation}\label{eq:max-equiv}
\max_{0\leq k\leq q^n:|\mathrm{M}^n_q(k)_{\leq d}|\leq m} k = \begin{cases}
\max_{0\leq k\leq q^n: |\mathrm{M}^n_q(k)_{\leq d}|=m} k & \text{if } m\leq N(n,d,q-1),\\
q^n & \text{otherwise.}
\end{cases}
\end{equation}
This is because if $|\mathrm{M}^n_q(k)_{\leq d}|<m$ and $k<q^n$, we may always increase $k$ by one, which increases $|\mathrm{M}^n_q(k)_{\leq d}|$ by at most one.

Let $T\subseteq\F_q^n$ be a set of size $m$. Let $U=\cl_d(T)$ and $k_U=|U|$. 
We want to prove $k_U\leq \max_{0\leq k\leq q^n:|\mathrm{M}^n_q(k)_{\leq d}|\leq m} k$.
Assume to the contrary that this is not true. Then $|\mathrm{M}^n_q(k_U)_{\leq d}|>m$.
So we have 
\begin{equation}\label{eq:hU-hT}
\h_U(d)\geq |\mathrm{M}^n_q(k_U)_{\leq d}|>m=|T|\geq \h_d(T)\;,
\end{equation}
where the first inequality holds by \cref{cor:boundsOnHilbert}.
On the other hand, the fact that $U=\cl_d(T)$ implies that $\h_U(d)=\h_T(d)$, contradicting \eqref{eq:hU-hT}. So $|\cl_d(T)|= k_U\leq \max_{0\leq k\leq q^n:|\mathrm{M}^n_q(k)_{\leq d}|\leq m} k$.
\end{proof}

The next theorem states that the bound in \cref{thm:tight-bound-cl} is tight and explicitly constructs sets that meet this bound.

\begin{theorem}\label{thm:closure-tightness}
Let $\sigma_A: \mathcal{M}\to A$ and $\phi: \mathcal{M}\to F$ be the bijections \eqref{eq:sigma-A} and \eqref{eq:phi} respectively, where $A=\F_q^n$, $F=\{0,1,\dots,q-1\}^n$, and $\mathcal{M}=\left\{\prod_{i=1}^n X_i^{e_i}: 0\leq e_1,\dots,e_n\leq q-1\right\}$.
Let $m$ be any integer such that $0\leq m\leq q^n$.
Choose the maximum $k\leq q^n$ such that $|\mathrm{M}^n_q(k)_{\leq d}|\leq m$.
Let $T_0=(\sigma_A\circ\phi^{-1})(\mathrm{M}_q^n(k)_{\leq d})\subseteq A=\F_q^n$.
If $|T_0|\geq m$, let $T=T_0$. Otherwise, let $T$ be an arbitrary set obtained by adding $m-|T_0|$ elements from $\F_q^n\setminus T_0$ to $T_0$. 
Then $T$ is a set of size $m$ that attains the equality in \eqref{eq:closure-size}.
\end{theorem}

\begin{proof}
Let $S=(\sigma_A\circ\phi^{-1})(\mathrm{M}_q^n(k))\supseteq T_0$.
Then $\phi(\sigma_A^{-1}(S))=\mathrm{M}_q^n(k)$.
By the definitions of $\mathrm{M}_q^n(k)$ and $\phi$, the set $\sigma_A^{-1}(S)=\phi^{-1}(\mathrm{M}_q^n(k))\subseteq \mathcal{M}$ is down-closed.
So by \cref{lem:A-to-F}, we have 
\[
\h_S(d)=|\phi(\sigma_A^{-1}(S))_{\leq d}|=|\mathrm{M}_q^n(k)_{\leq d}|\;.
\]
Similarly, by the definitions of $\mathrm{M}_q^n(k)_{\leq d}$ and $\phi$, the set $\sigma_A^{-1}(T_0)=\phi^{-1}(\mathrm{M}_q^n(k)_{\leq d})\subseteq \mathcal{M}$ is down-closed.
So by \cref{lem:A-to-F}, we have 
\[
\h_{T_0}(d)=|\phi(\sigma_A^{-1}(T_0))_{\leq d}|=|(\mathrm{M}_q^n(k)_{\leq d})_{\leq d}|=|\mathrm{M}_q^n(k)_{\leq d}|\;.
\]
It follows that $\h_S(d)=\h_{T_0}(d)$. So we have
$S\subseteq \cl_d(T_0)$. Therefore, 
\begin{equation}\label{eq:lower-bound-Cl}
|\cl_d(T)|\geq |\cl_d(T_0)|\geq |S|=|\mathrm{M}_q^n(k)|=k= \max_{0\leq k'\leq q^n: |\mathrm{M}^n_q(k')_{\leq d}|\leq m} k'\;,
\end{equation}
where the last equality holds by the choice of $k$.

Suppose $m>N(n,d,q-1)$. Then $k=q^n$. So $|T_0|=|\mathrm{M}_q^n(k)_{\leq d}|=N(n,d,q-1)<m$.
Then $|T|=m$ by definition. As $|\cl_d(T)|\geq k=q^n$, we must have $|\cl_d(T)|=q^n$. So $T$ attains the equality in \eqref{eq:closure-size}.

Now suppose $m\leq N(n,d,q-1)$.
Then $|T_0|=|\mathrm{M}^n_q(k)_{\leq d}|=m$ by \eqref{eq:max-equiv} and the maximality of $k$. So $T=T_0$ and $|T|=m$. Combining \eqref{eq:closure-size} and \eqref{eq:lower-bound-Cl} shows that \eqref{eq:closure-size} holds with equality.
%
\end{proof}



\section{Low-Degree Dispersers}\label{sec:dispersers}
In this section, we will show how to use \cref{thm:hilbertbound} to conclude the existence of low-degree dispersers for various families of sources. In \cref{sec:dispSmall}, we will use \cref{cor:boundsOnHilbert} to show that for every family of at most $2^{O(k^d)}$ sources of min-entropy $k$, there exists a degree-$d$ disperser. In particular, this will imply dispersers for local sources and bounded-depth decision forest sources. In \cref{sec:dispCkts}, we will extend this result to large families of sources, including polynomial and circuit sources.

\subsection{Dispersers for Small Families of Sources}\label{sec:dispSmall}
In \cref{thm:pr-nonconst}, we use the bound of \cref{cor:boundsOnHilbert} on the values of Hilbert functions to bound the probability that a random polynomial takes a fixed value on an arbitrary subset of $\F_2^n$. 

\begin{theorem}\label{thm:pr-nonconst}
Let $n,d\geq1$, $S\subseteq \F_2^n$ be an arbitrary nonempty set, and $f:S\rightarrow \F_2$ be a function. Then,
\begin{align}\label{eq:disperser}
\Pr_{p\in_u \Poly{2}{n,d}}[p\vert_S \equiv f] \leq 2^{-\h_S(d,\F_2)} \leq 2^{- \binom{\lfloor \log_2 |S|\rfloor}{\leq d}}\;.
\end{align}
\end{theorem}
\begin{proof}
Given a function $f$, let $V_f$ be the subset of $\Poly{2}{n,d}$ consisting of polynomials $p$ such that $p|_S\equiv f$. Note that $V_0$ is a subspace of $\Poly{2}{n,d}$, and for every $f:S\rightarrow \F_2$, $V_f$ is either empty or is a coset of the subspace $V_0$. Thus by definition, for every $f$, 
\[
\Pr_{p\in_u \Poly{2}{n,d}}[p\vert_S \equiv f]\leq \frac{|V_0|}{\left| \Poly{2}{n,d} \right|}=2^{-\h_S(d,\F_2)} \;. 
\]
The second inequality in \cref{eq:disperser} follows from \cref{cor:boundsOnHilbert}.
\end{proof}

We will now use \cref{thm:pr-nonconst} to prove the existence of low-degree dispersers for every small family of sources.
\begin{corollary}\label{cor:disp}
Let $n,d,k\geq1$, and $\mathcal{X}$ be a family of distributions of min-entropy $\geq k$ over $\{0,1\}^n$. 
Then a uniformly random polynomial $p\in \Poly{2}{n,d}$ is a disperser for $\mathcal{X}$ with probability at least 
\[
1-|\mathcal{X}|\cdot 2^{1-{\binom{k}{\leq d}}} \;. 
\]
\end{corollary}
\begin{proof}
Let $\mathbf{X}$ be a distribution from $\mathcal{X}$.
Since $\ent(\mathbf{X})\geq k$, we have that $|\mathrm{support}(\mathbf{X})|\geq 2^k$. 
By \cref{thm:pr-nonconst}, 
\[
\Pr_{p\in_u \Poly{2}{n,d}}[p\vert_{\mathrm{support}(\mathbf{X})} \text{ is constant}] \leq 2^{1-\binom{k}{\leq d}}\;.
\]
The corollary follows by applying the union bound over all $|\mathcal{X}|$ sources in $\mathcal{X}$.
\end{proof}


We will demonstrate two immediate applications of \cref{cor:disp} for the families of local and decision forests sources.

\begin{corollary}[Low-degree dispersers for local sources]\label{cor:localdisperser}
Let $1\leq \ell\leq d\leq n$ be integers. There exists $p\in \Poly{2}{n,d}$ that is a disperser 
\begin{itemize}
    \item for the family of $\ell$-local sources on $\{0,1\}^n$ with min-entropy $k> d(2^\ell n + 2\ell n \log n)^{1/d}$. 
    \item for the family of depth-$\ell$ decision forest sources on $\{0,1\}^n$ with min-entropy $k> d((\ell+\log n) 2^{\ell+1} n)^{1/d}$. 
\end{itemize}
\end{corollary} 
\begin{proof}
Let $\mathcal{X}_\ell$ denote the family of $\ell$-local sources on $\{0,1\}^n$. By \cref{prop:countingSources},  
\[
|\mathcal{X}_\ell|\leq 2^{2^\ell n + 2\ell n\log n}\;.
\]
Thus, by \cref{cor:disp}, as long as $\binom{k}{\leq d}>2^\ell n+2\ell n\log n+1$, there exists a degree-$d$ disperser for $\ell$-local sources of min-entropy $k$. For $k> d(2^\ell n + 2\ell n \log n)^{1/d}$, we have that
\[
\binom{k}{\leq d} \geq \binom{k}{d}+1 \geq (k/d)^d+1 > 2^\ell n+2\ell n\log n+1 \;.
\]
For depth-$\ell$ decision forest sources, by \cref{prop:countingSources}, the number of such sources on $\{0,1\}^n$ is at most 
\[
2^{(\ell+\log n) 2^{\ell+1} n}\;.
\]
By \cref{cor:disp}, as long as $\binom{k}{\leq d}>(\ell+\log n) 2^{\ell+1} n+1$, there exists a degree-$d$ disperser for depth-$\ell$ decision forest sources of min-entropy $k$. For $k> d((\ell+\log n) 2^{\ell+1} n)^{1/d}$, we have that
\[
\binom{k}{\leq d} \geq \binom{k}{d}+1 \geq (k/d)^d+1 > (\ell+\log n) 2^{\ell+1} n+1 \;.\qedhere
\]
\end{proof}

The recent result of \cite{alrabiah2022low} uses further properties of local sources to prove the existence of low-degree dispersers for local sources with min-entropy $k\geq c \ell^3 d\cdot (n\log n)^{1/d}$ for a constant $c>0$. Noting that every depth-$\ell$ decision forest source is also a $(2^\ell-1)$-local source, the disperser of \cite{alrabiah2022low} for local sources implies a result similar to the above.




\subsection{Dispersers for Polynomial and Circuit Sources}\label{sec:dispCkts}
In this section, we will extend the results of the previous section to prove the existence of low-degree dispersers for powerful families of sources including polynomial-size circuits and low-degree polynomial sources. Unlike the previous examples such as local sources, the sources considered here may non-trivially depend on an arbitrary number of inputs. For example, even a degree-$1$ (i.e. affine) source defined by an affine map $f:\F_2^m \rightarrow \F_2^n$ can depend on an arbitrary number $m\gg n$ of input bits. We get around this by restricting the map $f:\{0,1\}^m\to\{0,1\}^n$ defining the source to a low-dimensional affine subspace. Specifically, we will use the input-reduction procedure from~\cite{chattopadhyay2024extractors}, where it was used to prove that random (not necessarily bounded degree) maps extract from low-degree sources.

\begin{lemma}[{\cite[Lemma 4.5]{chattopadhyay2024extractors}}]\label{lem:inpRedForDisp}
Let $m,n,k\in\N$, $k>1$, and $f\colon\F_2^m\to\F_2^n$ be a function. If $\ent(f(\mathbf{U}_m))\geq k$, then there exists an affine map $L\colon\F_2^{11k}\to\F_2^m$ such that 
\[
\ent\left(f\left(L(\mathbf{U}_{11k})\right)\right)\geq k-1\;.
\]
\end{lemma}

Equipped with \cref{lem:inpRedForDisp}, we are ready to construct dispersers for low-degree sources. 

\begin{theorem}[Low-degree disperser for lower-degree polynomial sources]\label{thm:dispPolys}
Let $1\leq \ell< d\leq n$ be integers. There exists $p\in \Poly{2}{n,d}$ that is a disperser for the family of degree-$\ell$ sources on $\{0,1\}^n$ with min-entropy $k\geq (12^\ell\cdot d^d\cdot n)^{\frac{1}{d-\ell}}+1$. 

In particular, for every $\ell\in\N$, there is a degree-$(\ell+2)$ disperser for degree-$\ell$ sources on $\{0,1\}^n$ with min-entropy $\Omega\left(\sqrt{n}\right)$. 
\end{theorem}
\begin{proof}
Suppose $\mathbf{X}=f(\mathbf{U}_m)$ is a source of min-entropy $\ent(\mathbf{X})\geq k$ defined by a degree-$\ell$ map $f:\F_2^m \rightarrow \F_2^n$.

We first reduce the number of inputs of the polynomial~$f$. Specifically, we construct another degree-$\ell$ polynomial~$g$ with $r=11k$ inputs such that every disperser for~$g$ is also a disperser for~$f$. We will then conclude the proof by showing that there exists a degree-$d$ disperser for the class of sources defined by polynomials with $r$ inputs.

Let $L\colon\F_2^r\to\F_2^m$ be the affine map from \cref{lem:inpRedForDisp}. 
Consider the new source $\mathbf{Y}=f(L(U_r))$ over $\{0,1\}^n$. By \cref{lem:inpRedForDisp}, $\ent(\mathbf{Y})\geq k-1$, and the map $f\circ L$ defining $\mathbf{Y}$ is also a degree-$\ell$ polynomial with only $r$ input variables. 
Since $\mathrm{Support}(\mathbf{Y})\subseteq\mathrm{Support}(\mathbf{X})$, a disperser for $\mathbf{Y}$ is also a disperser for $\mathbf{X}$.

Therefore, it is now sufficient to prove that there exists a degree-$d$ polynomial $p$ that is a disperser for all degree-$\ell$ sources $\mathbf{Y}$ with $\ent(\mathbf{Y})\geq k-1$ defined by polynomials with $r$~inputs. By \cref{prop:countPolys}, the number of such sources is bounded from above by $2^{n\cdot\binom{r}{\leq \ell}}$. Now, \cref{cor:disp} guarantees the existence of a degree-$d$ disperser as long as
\[
{\binom{k-1}{\leq d}} -1 > n\cdot {\binom{r}{\leq \ell}} = n\cdot {\binom{11k}{\leq \ell}}\;.
\]
For $k\geq (12^\ell\cdot d^d\cdot n)^{\frac{1}{d-\ell}}+1$, we have that
\begin{align*}
    \binom{k-1}{\leq d} -1
    &> \binom{k-1}{d}\\
    &\geq \left((k-1)/d\right)^d\\
    &\geq (k-1)^\ell \cdot (k-1)^{d-\ell}/d^d\\
    &\geq (k-1)^\ell \cdot (12^\ell\cdot d^d\cdot n) /d^d\\
    &=n\left(12(k-1)\right)^\ell\\
    &\geq n\left((11k)^\ell+1\right)\\
    &\geq n \cdot \binom{11k}{\leq \ell} \;,
\end{align*}
which finishes the proof of the theorem.
\end{proof}

\begin{theorem}[Low-degree disperser for circuit sources]\label{thm:dispCkts}
Let $\ell\geq1$ and $n\geq d\geq 2\ell+2$ be integers. There exists $p\in \Poly{2}{n,d}$ that is a disperser for the family of $n^\ell$-size circuit sources on $\{0,1\}^n$ with min-entropy $k\geq (30^2\cdot d^d\cdot n^{2\ell})^{\frac{1}{d-2}}+1$.
\end{theorem}

\begin{proof}
    Let $\mathbf{X}=C(\mathbf{U}_m)$ be a source of min-entropy $\ent(\mathbf{X})\geq k$ defined by an $\ACp$ circuit, $C:\F_2^m \rightarrow \F_2^n$ of size $n^\ell$. 

Similarly to the proof of \cref{thm:dispPolys}, we first reduce the number of inputs of the circuit~$C$. Let $L\colon\F_2^r\to\F_2^m$ be the affine map from \cref{lem:inpRedForDisp} for $r=11k$. Consider the circuit source $\mathbf{Y}=C(L(\mathbf{U}_r))$.

By \cref{lem:inpRedForDisp}, $\ent(\mathbf{Y})\geq k-1$, and since $\mathrm{Support}(\mathbf{Y})\subseteq\mathrm{Support}(\mathbf{X})$, a disperser for $\mathbf{Y}$ is also a disperser for $\mathbf{X}$.

We will now show that there is an $\ACp$ circuit $C'\colon\F_2^r\to\F_2^n$ of $\size(C')\leq rn^\ell$ computing $C'(y)=C(L(y))$ for every $y=(y_1,\ldots,y_r)\in\F_2^r$. The main difference between the circuits $C$ and $C'$ is that we need to replace each input $x_i$ for $i\in[m]$ of the circuit $C$ by an affine form $L_i(y_1,\ldots,y_r)$ of the inputs of~$C'$ given by~$L$. The naive way of implementing this modification would require $m\gg n^\ell$ additional gates which we cannot afford. Instead, we will simulate this modification by updating the input wires of each gate in~$C'$ accordingly (and introducing at most $r$ additional gates for each gate in~$C$). 

Every $\XOR$ gate of~$C$ fed by inputs $x_{i_1},\ldots,x_{i_t}$ will now be replaced by an $\XOR$ gate in~$C'$ computing $L_{i_1}(y_1,\ldots,y_r)\oplus\ldots\oplus L_{i_t}(y_1,\ldots,y_r)$. Note that this operation only modifies the wires in the circuit and does not introduce new gates. Every $\AND$ gate of~$C$ fed by inputs $x_{i_1},\ldots,x_{i_t}$ will now be replaced by a gate computing $\AND$ of at most $r$ linearly independent forms from $L_{i_1},\ldots,L_{i_t}$. This operation introduces at most $r$ additional $\XOR$ gates for each gate of the original circuit~$C$. We handle $\OR$ gates in a similar fashion. Finally, the wires between non-input gates of the circuit~$C$ stay unchanged in the circuit~$C'$. We constructed a circuit of size $rn^\ell$ with $r$ inputs that defines the source~$\mathbf{Y}$.

It remains to show the existence of a degree-$d$ polynomial $p$ that is a disperser for all sources $\mathbf{Y}$ with $\ent(\mathbf{Y})\geq k-1$ defined by $rn^\ell$-size circuits with $r$~inputs. By \cref{prop:countPolys}, the number of such sources is bounded from above by 
\[
2^{4rn^\ell(rn^\ell+r)}
=
2^{4r^2n^\ell(n^\ell+1)}
= 
2^{4(11k)^2 n^\ell(n^\ell+1)}
\leq
2^{(27k)^2 n^{2\ell}}
\;.
\]
Now, \cref{cor:disp} guarantees the existence of a degree-$d$ disperser as long as
\[
{\binom{k-1}{\leq d}} -1 > (27k)^2 n^{2\ell}\;.
\]
For $k\geq (30^2\cdot d^d\cdot n^{2\ell})^{\frac{1}{d-2}}+1$, we have that
\begin{align*}
\binom{k-1}{\leq d} -1 
&> \binom{k-1}{d} \\
&\geq \left((k-1)/d\right)^d\\
&\geq (k-1)^2\cdot (k-1)^{d-2} / d^d\\
&\geq (k-1)^2 \cdot (30^2\cdot d^d\cdot n^{2\ell}) / d^d\\
&\geq (30(k-1))^2 n^{2\ell}\\
&\geq (27k)^2 n^{2\ell} \;,
\end{align*}
which finishes the proof.
\end{proof}

\cref{thm:dispPolys,thm:dispCkts} construct low-degree dispersers for sources generated by constant-degree polynomials and polynomial-size $\ACp$ circuits. These two classes of sources are incomparable. Indeed, $\ACp$ computes $\AND(x_1,\ldots,x_m)$ which is not a constant-degree polynomial, while constant-degree polynomials compute polynomials in~$m$ inputs which do not admit circuits of size polynomial in~$n$. We remark that the techniques of \cref{thm:dispPolys,thm:dispCkts} can be used to conclude the same result for a class of sources that generalizes both $\ACp$ and constant-degree polynomials. This is the class of polynomial-size circuits which extends $\ACp$ with gates computing arbitrary polynomials in~$m$ inputs of a fixed constant degree. For ease of exposition, we present only the results for more natural sources in \cref{thm:dispPolys,thm:dispCkts}.

\section{Random Low-Degree Polynomials Extract from Fixed Sources}
In this section, we use our bounds on the values of Hilbert functions to prove the existence of a low-degree extractor for a fixed high min-entropy source. Specifically, in \cref{thm:main-extractor} we show that for every source~$\mathbf{X}$ of high min-entropy, a random low-degree polynomial~$p$ has bias $\leq\eps$, i.e., $\Pr_{x\in_u X}[f(x)=1]\in 1/2\pm \eps$ with high probability. One special case of interest is the case of $k$-\emph{flat} sources $\mathbf{X}$ which are uniform distributions over sets of size $2^k$. In \cref{sec:extractors}, we will use \cref{thm:main-extractor} to prove the existence of low-degree extractors for various expressive families of sources.

We start this section by using our bounds on the degree-$d$ closure of sets in order to lower-bound the probability that a random somewhat large subset $T$ of a set $S$ has ``full Hilbert dimension'', i.e., $\h_T(d, \F_2)=|T|$. We then use this to prove \cref{clm:probability-x-in-good-subset} which states that for a large enough set~$S\subseteq\{0,1\}^n$, a random subset $T\subseteq S$ of full Hilbert dimension will contain each element $x\in S$ with almost the same probability. Finally, we present a proof of \cref{thm:main-extractor} which crucially relies on \cref{clm:probability-x-in-good-subset}.

\begin{claim}\label{clm:probability-good subset}
Let $1\leq d\leq n$, $d\leq \ell$, and $S\subseteq\{0,1\}^n$. Let $T$ be a uniformly random subset of $S$ of size $\binom{\ell}{\leq d}$. Then 
\[
\Pr_{T}\left[\h_T(d, \F_2)=|T|=\binom{\ell}{\leq d}\right] \geq 1-\binom{\ell}{\leq d}\cdot2^{\ell}/|S|\;.
\]
\end{claim}
\begin{proof}
    For a random subset $T=\left\{a_1,a_2,\dots,a_{\binom{\ell}{\leq d}}\right\}$ and $j\in\left[\binom{\ell}{\leq d}\right]$, let $T^j=\left\{a_1,\dots,a_j\right\}$. We~prove that with high probability, 
    $\h_{T^{j+1}}(d, \F_2)=\h_{T^{j}}(d, \F_2)+1$ holds for every~$j$. Indeed, for every $j<\binom{\ell}{\leq d}$, we have that
\begin{align*}
\Pr\left[\h_{T^{j+1}}(d, \F_2)=\h_{T^{j}}(d,\F_2)+1\right]
        &=\Pr\left[a_{j+1}\notin \cl_d(T^j)\right]\\
        &=1-|\cl_d(T^j)|/|S|\\
        &\geq 1-2^{\ell}/|S| \;,
\end{align*}
where the last inequality uses \cref{cor:bound-span}. The claim now follows by the union bound over all~$j$. 
\end{proof}

\begin{lemma}\label{clm:probability-x-in-good-subset}
Let $1\leq d\leq n$, $d\leq \ell$, and $S\subseteq\{0,1\}^n$. Let $T$ be a uniformly random subset of $S$ of size $\binom{\ell}{\leq d}$. Then for every $x\in S$,
\[
(1-\delta)\cdot\frac{\binom{\ell}{\leq d}}{|S|}
\leq
\Pr_T\left[x\in T \mid \h_T(d,\F_2)=|T|\right]
\leq
\frac{1}{(1-\delta)}\cdot\frac{\binom{\ell}{\leq d}}{|S|}\;,
\]
where $\delta=\binom{\ell}{\leq d}\cdot2^{\ell}/|S|$.
\end{lemma}
\begin{proof}
By the Bayes rule, we have 
\begin{align*}\Pr_T[x\in T \mid \h_T(d,\F_2)=|T|]
&=\frac{\Pr[\h_T(d,\F_2)=|T| \mid x\in T]\cdot\Pr[x\in T]}{\Pr[\h_T(d,\F_2)=|T|]}\\
&=\frac{\Pr[\h_T(d,\F_2)=|T| \mid x\in T]}{\Pr[\h_T(d,\F_2)=|T|]}\cdot
\frac{\binom{\ell}{\leq d}}{|S|}
\;.
\end{align*}

For the upper bound, by \cref{clm:probability-good subset}, $\Pr[\h_T(d,\F_2)=|T|]\geq 1- \binom{\ell}{\leq d}\cdot2^{\ell}/|S|=1-\delta$. Thus, 
\[
\Pr[x\in T\ |\ \h_T(d,\F_2)=|T|]\leq \frac{1}{(1-\delta)}\cdot\frac{\binom{\ell}{\leq d}}{|S|}\;.
\]

For the lower bound, first include $x$ in~$T$, and then randomly pick the remaining $\binom{\ell}{\leq d}-1$ elements in~$T$. Analogous to the analysis in \cref{clm:probability-good subset}, we have 
$\Pr_T[\h_T(d,\F_2)=|T|\ |\ x\in T]\geq 1- \binom{\ell}{\leq d}\cdot2^{\ell}/|S|=1-\delta$.
Thus, 
\[
\Pr_T[x\in T\ |\ \h_T(d,\F_2)=|T|] \geq (1-\delta)\cdot\frac{\binom{\ell}{\leq d}}{|S|} \;,
\]
which concludes the proof.
\end{proof}

Equipped with \cref{clm:probability-x-in-good-subset}, we are ready to present the proof of \cref{thm:main-extractor}.
\begin{theorem}\label{thm:main-extractor}
Let $n,d,k\geq1$, and $\eps>0$ be a real. Then for every distribution $\mathbf{X}$ over $\{0,1\}^n$ with $\ent(\mathbf{X})\geq k$,  a uniformly random degree-$d$ polynomial $f$ is an $\eps$-extractor for $\mathbf{X}$,
    \[ 
\Pr_{x\sim \mathbf{X}}[f(x)=1] = \frac{1}{2}\pm \eps\;
    \]
with probability at least $1-e^{3n-\eps^2\binom{\ell}{\leq d}/(Cn^2)}$
where $\ell=k/2-\log(32n/\eps)$ and $C=7\cdot(32)^2$.
\end{theorem}

\begin{proof}
We will assume that  $\ell\coloneqq k/2-\log(32n/\eps)>0$ (and, in particular, that $\eps\geq2^{-n}$) as otherwise the theorem statement holds trivially. Similarly, we will assume that $3n-\eps^2\binom{\ell}{\leq d}/(Cn^2)<0$, and, in particular, that $2^{n+1} \cdot e^{-(\eps/(32n))^2 \cdot{\binom{\ell}{\leq d}}/7}<1$.

Let $\mathbf{X}$ be a distribution over $\{0,1\}^n$ with $\ent(\mathbf{X})\geq k$, and let $S=\mathrm{support}(\mathbf{X})\subseteq \{0,1\}^n$. For $x\in S$, let $p_x=\Pr[\mathbf{X}=x]$, so that $p_x\leq 2^{-k}$ and $\sum_{x\in S}p_x = 1$. For every $i\geq k$, define 
\[
S_i\coloneqq \{x\in S : 2^{-i-1}< p_x\leq 2^{-i}\}\;,
\]
and note that $|S_i|\leq 2^{i+1}$. 
We will show that with high probability a random degree-$d$ polynomial~$f$ satisfies 
\begin{align}\label{eq:mainEqExtr}
\sum_{i\geq k}\left|\sum_{x\in S_i} p_x f(x)-\frac{1}{2}\sum_{x\in S_i} p_x\right| \leq \eps \;.
\end{align}
Given \cref{eq:mainEqExtr}, we finish the proof of the theorem as follows.
\begin{align*}
\left|\Pr_{x\sim X}[f(x)=1]-1/2\right| 
&= 
\left|\sum_{x\in S} p_x f(x)-1/2\right|\\
&=\left|\sum_{x\in S} p_x f(x)-\frac{1}{2}\sum_{x\in S} p_x\right|\\
&\leq \sum_{i\geq k}\left|\sum_{x\in S_i} p_x f(x)-\frac{1}{2}\sum_{x\in S_i} p_x\right| \\
&\leq \eps
\;.
\end{align*}
It remains to prove \cref{eq:mainEqExtr}, and we will do this by analyzing the terms for different values of~$i$ separately. 

\paragraph{Large $i$.} Let $S_L$ be the union of all $S_i$ with $i\geq 2n+1$. Note that 
\[
\sum_{x\in S_L} p_x \leq 2^n \cdot 2^{-2n-1} \leq \eps/2 \;, 
\]
where we used that $\eps\geq 2^{-n}$.
Thus, in the rest of the proof we will only deal with $k\leq i\leq 2n$. Specifically, we will show that for each $i$ in this interval,
$\left|\sum_{x\in S_i} p_x f(x)-\frac{1}{2}\sum_{x\in S_i} p_x\right|\leq 8\eps'$ for 
\[
\eps'\coloneqq \frac{\eps}{32n}\;.
\]
We will assume that $\eps\leq 1/2$ (as the theorem statement is trivial otherwise), and, thus, $\eps'\leq 1/64$.

\paragraph{Small $S_i$.} 
If $|S_i|\leq 2^{2\ell+1}/\eps'$, then 
\[
\sum_{x\in S_i}p_x \leq 2^{2\ell+1-i}/\eps' \leq 2^{2\ell+1-k}/\eps'\leq \eps'\;,
\]
where we used that $i\geq k$ and $\ell= k/2-\log(32n/\eps)=k/2-\log(1/\eps')$.

\paragraph{Small~$i$ and large~$S_i$.} Let $I$ be the set of $i\leq 2n$ such that $|S_i|> 2^{2\ell+1}/\eps'$. 

Fix an $i\in I$ and let $t=|S_i|$. We independently and uniformly at random pick subsets $T_1,\dots, T_t$ of $S_i$ of size $|T_j|=\binom{\ell}{\leq d}$ satisfying $\h_{T_j}(d,\F_2)=|T_j|$. For any $x\in S_i$, let $n_x$ denote the number of sets $T_j$ containing $x$: $n_x=|\{j\in[t]\colon x\in T_j\}|$. Note that by \cref{clm:probability-x-in-good-subset}, 
\[
(1-\delta)\cdot \binom{\ell}{\leq d} \leq \E[n_x] \leq \frac{1}{(1-\delta)}\cdot \binom{\ell}{\leq d}
\;,\]
where $\delta=\binom{\ell}{\leq d}\cdot2^{\ell}/|S_i|$. Since $|S_i|> 2^{2\ell+1}/\eps'$, we have that $\delta \leq\eps'/2\leq 1/2$, and
\[\E[n_x] = (1\pm\eps')\cdot\binom{\ell}{\leq d}\;.\]

Applying the Chernoff bound for the concentration of $n_x$ and the union bound over all $t$ choices of $x\in S_i$, we have that
\begin{align*}
\Pr_{T_1,\dots, T_t} \left[ \exists x\in S_i, \ n_x\notin (1\pm 2\eps')\cdot\binom{\ell}{\leq d}\right]\leq  2 t \cdot e^{-\eps'^2 \cdot{\binom{\ell}{\leq d}}/7} \leq 2^{n+1} \cdot e^{-(\eps/(32n))^2 \cdot{\binom{\ell}{\leq d}}/7} <1\;.
\end{align*}
In particular, there exists a choice of $T_1,\ldots,T_t$ such that $\h_{T_j}(d,\F_2)=|T_j|$ for all $j\in[t]$, and each $n_x=(1\pm 2\eps')\cdot\binom{\ell}{\leq d}$. Let us fix this choice of sets $T_1,\ldots,T_t$.

For a fixed set $T_j$ such that $\h_{T_j}(d,\F_2)=|T_j|$, a random degree-$d$ polynomial~$f$ satisfies
\begin{align}\label{eq:fullHDim}
\Pr_{f\in_u \Poly{2}{n,d}}\left[\left| \sum_{x\in T_j} p_x f(x) - \frac{1}{2}\sum_{x\in T_j}  p_x \right|\geq \frac{\eps' \binom{\ell}{\leq d}}{2^i}\right]\leq 2e^{-2\eps'^2 \binom{\ell}{\leq d}} \;.
\end{align}
Indeed, $\h_{T_j}(d,\F_2)=|T_j|$ implies that a random degree-$d$ polynomial~$f$ induces a random map  $f|_{T_j}\colon T_j\rightarrow \{0,1\}$ (see, e.g., the proof of \cref{thm:pr-nonconst}). \cref{eq:fullHDim} now follows from the Hoeffding bound, as the random variables $2^i p_xf(x)$ are $\binom{\ell}{\leq d}$ independent $[0,1]$-valued random variables with mean $2^{i-1} p_x$. 

Taking the union bound over all $j\in[t]$, we have that 
\begin{align*}
\Pr_{f\in_u \Poly{2}{n,d}}\left[\forall j\in[t], \ \left| \sum_{x\in T_j} p_x f(x) - \frac{1}{2}\sum_{x\in T_j}  p_x \right|\leq \frac{\eps' \binom{\ell}{\leq d}}{2^i}\right]\geq 1-2te^{-2\eps'^2 \binom{\ell}{\leq d}} 
\;.
\end{align*}

Therefore, using $t\leq 2^n$, with probability at least 
\[
1-2^{n+1}\cdot e^{-2\eps'^2 \cdot{\binom{\ell}{\leq d}}}\;,
\] we simultaneously have that 
(i) for every $x\in S_i$, $n_x=(1\pm 2\eps')\cdot\binom{\ell}{\leq d}$, and (ii) for every $j\in[t]$,
\[
\left| \sum_{x\in T_j} p_x f(x) - \frac{1}{2}\sum_{x\in T_j}  p_x \right|\leq \frac{\eps' \binom{\ell}{\leq d}}{2^i}\;.
\]
Conditioning on this good event, we get

\begin{align*}
\sum_{x\in S_i} p_x f(x) &= \sum_{x\in S_i} \frac{n_x}{n_x} p_x f(x)\\ &=
\frac{1}{(1\pm 2\eps')\cdot\binom{\ell}{\leq d}}\cdot \sum_{j\in [t]}\sum_{x\in T_j} p_x f(x)\\ &=
\frac{1}{(1\pm 2\eps')\cdot\binom{\ell}{\leq d}}\cdot \sum_{j\in [t]}\left(\left(\frac12\sum_{x\in T_j} p_x\right) \pm \frac{\eps' \binom{\ell}{\leq d}}{2^i}\right)
\\ &=
\frac{1}{(1\pm 2\eps')\cdot\binom{\ell}{\leq d}}\cdot \left(\frac12\sum_{x\in S_i}n_xp_x\right)\pm \frac{\eps' t}{(1\pm 2\eps') \cdot2^i}
\\ &=
\frac{(1\pm 2\eps')}{(1\pm 2\eps')}\cdot \left(\frac12\sum_{x\in S_i}n_xp_x\right)\pm \frac{2\eps'}{(1\pm 2\eps')}
\\ &=
\left(\frac{1}{2}\sum_{x\in S_i}p_x \right)\pm 8\eps',
\end{align*}
where the penultimate equality uses $t=|S_i|\leq 2^{i+1}$.

Finally, we apply the union bound over all $i\in I$. Since $|I|\leq 2n$, combining the contributions from small $S_i$'s and large $i$'s, we conclude  that 
\[
\left| \frac{1}{2}\sum_{x\in S_i}p_x - \sum_{x\in S_i} p_x f(x)\right| \leq \frac{\eps}{2}+ 8\eps'\cdot|I|\leq \eps\;
\]
with probability at least
\[
1-n\cdot 2^{n+2} \cdot e^{-2\eps'^2 \cdot{\binom{\ell}{\leq d}}}
\geq 1 - e^{3n-\eps^2 \cdot{\binom{\ell}{\leq d}}/(512n^2)}
\geq 1 - e^{3n-\eps^2 \cdot{\binom{\ell}{\leq d}}/(Cn^2)}
\;.\qedhere
\]
\end{proof}

\section{Low-Degree Extractors}\label{sec:extractors}
In this section, we extend the results of \cref{sec:dispersers} to the setting of extractors. We start with the extractors version of \cref{cor:disp} in \cref{thm:extr}, where we show that low-degree polynomials extract from small families of sources. Then, in \cref{thm:extrForSources}, we use \cref{thm:extr} to prove the existence of low-degree extractors for a number of families of sources. Finally, in \cref{sec:multBits}, we prove the existence of low-degree extractors with multi-bit outputs.

\begin{theorem}\label{thm:extr}
Let $\mathcal{X}$ be a family of distributions of min-entropy $k\geq5\log{n}$ over $\{0,1\}^n$ for large enough~$n$. Let $\mathcal{Y}$ be a family of distributions each of which is $\eps'$-close to a convex combination of distributions from~$\mathcal{X}$.
Then for every $d\geq6$, a uniformly random polynomial $p\in \Poly{2}{n,d}$ is an $\eps$-extractor for $\mathcal{Y}$ with probability at least 
\[
1-|\mathcal{X}|\cdot e^{3n-30k^{d/2}/n^2}\;
\]
for $\eps=\left(2d/k^{1/4}\right)^d+\eps'$.
\end{theorem}
\begin{proof}
Let $\delta=\eps-\eps'=\left(2d/k^{1/4}\right)^d$. We first bound the probability that a random polynomial is a $\delta$-extractor for a fixed source $\mathbf{X} \in \mathcal{X}$, and then apply the union bound over all sources from $\mathcal{X}$.  

Let $\ell=k/2-\log(32n/\delta)\geq k/3$ for all large enough~$n$. Let $C=7\cdot(32)^2$.
By \cref{thm:main-extractor}, we have that the probability that a random degree-$d$ polynomial is \emph{not} a $\delta$-extractor for a fixed source $\mathbf{X}\in\mathcal{X}$ is at most 
\begin{align*}
e^{3n-\delta^2\binom{\ell}{\leq d}/(Cn^2)}
&\leq e^{3n - (2d)^{2d}\cdot k^{-d/2}\cdot(\ell/d)^d/(Cn^2)}\\
& \leq e^{3n-30k^{d/2}/n^2}\;.
\end{align*}
Now the union bound over all sources from $\mathcal{X}$ gives us that a random degree-$d$ polynomial is a $\delta$-extractor for every source in $\mathcal{X}$ with probability at least
\[
1-|\mathcal{X}|\cdot e^{3n-30k^{d/2}/n^2}\;.
\] 
Finally, each polynomial that $\delta$-extracts from $\mathcal{X}$ is also an $\eps$-extractor for all sources in $\mathcal{Y}$.
\end{proof}

We will use the following input-reduction result from~\cite{chattopadhyay2024extractors}.

\begin{theorem}[{\cite[Theorem 4.1]{chattopadhyay2024extractors}}]\label{thm:inputRedExtr}
Let $m,n,k\in\N$, $k>1$, and $f\colon\F_2^m\to\F_2^n$ be a function. If $\ent(f(\mathbf{U}_m))\geq k$, then there exist affine maps $L_1,\ldots,L_t\colon\F_2^{11k}\to\F_2^m$ such that the distribution  $f(\mathbf{U}_m)$ is $2^{-k}$-close to a convex combination of distributions $f\left(L_i(\mathbf{U}_{11k})\right)$. Moreover, for each $i\in[t]$,
\[
\ent\left(f\left(L_i(\mathbf{U}_{11k})\right)\right)\geq k-1\;.
\]
\end{theorem}

We are now ready to prove that low-degree polynomials extract from many sources of interest.
\begin{theorem}\label{thm:extrForSources}
For all $\ell,d\geq1$, and all large enough~$n$, there exists $p\in \Poly{2}{n,d}$ that is an $\eps$-extractor for the following families of sources over $\{0,1\}^n$ of min-entropy $k\geq5\log{n}$ for ${\eps=2\left(2d/k^{1/4}\right)^d}$.
\begin{itemize}
\item $\ell$-local sources for $k\geq (2^\ell n^3\log{n})^{2/d}$.
\item depth-$\ell$ decision forest sources for $k\geq (2^\ell n^3(\log{n}+\ell))^{2/d}$.
\item degree-$\ell$ sources for $k\geq(3^\ell n)^{\frac{6}{d-2\ell}}$.
\item $n^\ell$-size circuit sources for $k\geq 3n^{\frac{4(\ell+1)}{d-4}}$.
\end{itemize}
\end{theorem}
\begin{proof}
By \cref{prop:countingSources}, the number of $\ell$-local sources over $\{0,1\}^n$ is at most $2^{2^\ell n + 2\ell n\log n}$, and the number of depth-$\ell$ decision forest sources is at most $2^{(\ell+\log n)2^{\ell+1} n}$. 
Now \cref{thm:extr} implies the result for these classes of sources.

Let $\mathcal{Y}$ be the family of degree-$\ell$ sources. We first apply \cref{thm:inputRedExtr}, and obtain the family $\mathcal{X}$ of degree-$\ell$ sources such that each source in $\mathcal{Y}$ is $2^{-k}$-close to a convex combination of sources in $\mathcal{X}$. Moreover, each source in $\mathcal{X}$ is a degree-$\ell$ polynomial in $11k$ variables and has entropy $\geq k-1$.
By \cref{prop:countPolys}, the number of such sources is bounded from above by $|\mathcal{X}|\leq 2^{n\cdot\binom{11k}{\leq \ell}}$. Now, \cref{thm:extr} guarantees the existence of a degree-$d$ $\eps'$-extractor for $\mathcal{Y}$ for $k\geq(3^\ell n)^{\frac{6}{d-2\ell}}$ and $\eps'=\left(2d/k^{1/4}\right)^d+2^{-k}\leq\eps$, where the last inequality uses $d\leq k^{1/4}$.

When $\mathcal{Y}$ is the family of $n^\ell$-size circuit sources, using \cref{thm:inputRedExtr} and the argument from \cref{thm:dispCkts}, we obtain the family $\mathcal{X}$ of $(11k n^\ell)$-size circuits with $11k$ inputs and entropy $k-1$. Moreover, each source from $\mathcal{Y}$ is $2^{-k}$-close to a convex combination of sources from $\mathcal{X}$. By \cref{prop:countPolys}, $|\mathcal{X}|\leq 2^{(27k)^2 n^{2\ell}}$. Now \cref{thm:extr} guarantees the existence of a degree-$d$ $\eps'$-extractor for $\mathcal{Y}$ for $k\geq 3n^{\frac{4(\ell+1)}{d-4}}$ and $\eps'=\left(2d/k^{1/4}\right)^d+2^{-k}\leq\eps$. 
\end{proof}

\subsection{Extractors Outputting Multiple Bits}\label{sec:multBits}
In \cref{thm:extr-multi-output}, we show how to extend our single-bit extractors for small families of sources to the multi-bit setting, which combined with input-reduction lemma, will extend all our single-bit extractors from \cref{thm:extrForSources} to $O(k)$-bit extractors.

\begin{theorem}\label{thm:extr-multi-output}
Let $\mathcal{X}$ be a family of distributions of min-entropy $k\geq5\log{n}$ over $\{0,1\}^n$ for large enough~$n$. Let $\mathcal{Y}$ be a family of distributions each of which is $\eps'$-close to a convex combination of distributions from~$\mathcal{X}$.
Then for every $d\geq 6$ and $t<k$, let $p_1,\dots, p_t \in \Poly{2}{n,d}$ be independent and uniformly random polynomials. Then $p=(p_1,\ldots, p_t)$ is a $t\eps$-extractor for $\mathcal{Y}$ with probability at least 
\[
1-|\mathcal{X}|\cdot e^{3n+t+1-30(k-2t)^{d/2}/n^2}\;
\]
for $\eps=\left(2d/k^{1/4}\right)^d+\eps'$, assuming $\eps\leq 1/4$. 
\end{theorem}
\begin{proof}
Define $\mathcal{X}_i$ to be the family of sources resulting by conditioning sources $\mathbf{X}\in \mathcal{X}$ on $(p_1,\ldots, p_i)= (b_1,\ldots, b_i)\in \F_2^n$ for any $(b_1,\ldots, b_i)\in \{0,1\}^i$, so that $\mathcal{X}_0=\mathcal{X}$, and $|\mathcal{X}_i|=2^i |\mathcal{X}|$. Let $E_i$ be the event that $p_i$ is an $\eps$-extractor for $\mathcal{X}_i$, and let $E_{\leq i}=E_1\wedge \dots \wedge E_i$. 

Note that, conditioned on $E_{\leq i}$, every source in $\mathcal{X}_{i+1}$ consists of $\leq 2^{i+1}|\mathcal{X}|$ sources of min-entropy $\geq k-2i$. This can be shown by an induction. The base case is true since $\mathcal{X}_1=\mathcal{X}$, and $\mathcal{X}$ has min-entropy $\geq k$. For the inductive step, let $\mathbf{X}'\in \mathcal{X}_{i+1}$ be obtained by conditioning $\mathbf{X}\in \mathcal{X}_i$ on $p_i=b$ for some $b\in \{0,1\}$. Now by the Bayes rule, we have for every $x$ in the support of $\mathbf{X}'$
\[
\Pr[\mathbf{X}'=x] = \Pr[\mathbf{X}=x \vert p_i(X)=b] = \frac{\Pr[p_i(\mathbf{X})=b\vert \mathbf{X}=x]\cdot \Pr[\mathbf{X}=x]}{\Pr[p_i(\mathbf{X})=b]}\leq \frac{2^{-(k-2(i-1))}}{\frac{1}{2}-\epsilon} \leq 2^{-(k - 2i)},
\]
as long as $\eps\leq \frac{1}{4}$, which shows that $\mathbf{X}'$ has min-entropy at least $k-2i$ as desired.

Thus, by \cref{thm:extr}, 
\[
\Pr[E_{i+1} \vert E_{\leq i}] \geq  1-|\mathcal{X}_{i+1}|\cdot e^{3n-30(k-2t)^{d/2}/n^2} \geq 1-2^{i+1}|\mathcal{X}|\cdot e^{3n-30(k-2t)^{d/2}/n^2}\;.
\]
Thus by the chain rule, we have
\begin{align*}
\Pr[E_1\wedge \cdots \wedge E_t] &= \prod_{i=0}^{t-1} \Pr[E_{i+1} \vert E_{\leq{i}}]\\ 
&\geq \prod_{i=0}^{t-1} \left(1-2^{i+1}|\mathcal{X}|\cdot e^{3n-30(k-2t)^{d/2}/n^2} \right)\\
&\geq 1- \sum_{i=0}^{t-1} \left(2^{i+1}|\mathcal{X}|\cdot e^{3n-30(k-2t)^{d/2}/n^2} \right)\\
&\geq 1- |\mathcal{X}|\cdot e^{t+1+3n-30(k-2t)^{d/2}/n^2}\;.
\end{align*}
Next, we analyze the error of the extractor conditioned on the above event $E_{\leq t}$. Let $\mathbf{X}\in \mathcal{X}$ be a fixed source. Let $(u_1,u_2,\ldots, u_t)$ be uniformly distributed over $\{0,1\}^t$ independent of $\mathbf{X}$ and $p_1,\ldots, p_t$, and define random variables $D_0, \ldots, D_t$ as follows. For every $i$, define 
\[
D_i=\left(p_1(\mathbf{X}), p_2(\mathbf{X}),\dots, p_i(\mathbf{X}), u_{i+1}, \ldots, u_t\right). 
\]
Note that $D_0$ is the uniform distribution and $D_t=p(\mathbf{X})$ is the output of the random degree-$d$ polynomial conditioned on $E_{\leq t}$. By the triangle inequality, we have
\[
\Delta(D_0, D_t) \leq \sum_{i=0}^{t-1} \Delta(D_i, D_{i+1}).
\]
Thus it suffices to bound each $\Delta(D_i, D_{i+1})$ by $\epsilon$, which can be obtained by noting that conditioning $\mathbf{X}$ on any value of $f_1,\ldots, f_{i}$ results in a source that belongs to $\mathbf{X}_{i}$ for which $f_{i+1}$ is an $\epsilon$-extractor. Since $\Delta(D_i, D_{i+1})$ is a convex combination of the bias of $f_{i+1}$ for such fixings, we obtain $\Delta(D_i, D_{i+1})\leq \epsilon$ as desired. 
\end{proof}

\bibliographystyle{alpha}
\bibliography{main}

\end{document}